\documentclass{amsart}

\usepackage{graphicx,bbm,txfonts,color}
\author{Patrik V. Nabelek}
\address{Department of Mathematics, Kidder Hall 368, Oregon State University, Corvallis, OR 97331-4605}
\email{nabelekp@oregonstate.edu}

\author{Vladimir E. Zakharov}
\address{Department of Mathematics, University of Arizona, Tucson, AZ 85721 \\ L.~D.~Landau Institute for Theoretical Physics, Chernogolovka, Moscow Region, Russian Federation }
\email{zakharov@math.arizona.edu}

\date{\today}

\title[Kaup--Broer System]{Solutions to the Kaup--Broer System and Its 2+1 Dimensional Integrable Generalization via the Dressing Method}

\newtheorem{thm}{Theorem}
\newtheorem{cor}[thm]{Corollary}
\newtheorem{rmk}[thm]{Remark}
\newtheorem{prop}[thm]{Proposition}
\newtheorem{rhp}[thm]{Riemann--Hilbert Problem}

\begin{document}

\maketitle

\begin{abstract}

In this paper we formulate the nonlocal dbar problem dressing method of Manakov and Zakharov \cite{ZM84,ZM85,Z89} for the 4 scaling classes of the 1+1 dimensional Kaup--Broer system \cite{B75,K75}.
The method for the 1+1 dimensional Kaup--Broer systems are reductions of a method for a complex valued 2+1 dimensional completely integrable partial differential equation first introduced in \cite{RP11}.
This method allows computation of solutions to all cases of the Kaup--Broer system.
We then consider the case of non-capillary waves with usual gravitational forcing, and use the dressing method to compute N-soliton solutions and more general solutions in the closure of the N-soliton solutions in the topology of uniform convergence in compact sets called primitive solutions.
These more general solutions are an analogue of the solutions derived in \cite{DZZ16,ZDZ16,ZZD16} for the KdV equation.
We derive dressing functions for finite gap solutions.
We compute counter propagating dispersive shockwave type solutions numerically.

\end{abstract}

\section{Introduction}

It was recently demonstrated experimentally by Redor et. al. that a soliton gas consisting of many solitons propagating in both directions can be formed in a flume \cite{RNMOM19}.
Moreover, they demonstrated that the counter propagating Kaup--Broer (Kaup--Boussinesq) 2-soliton soliton collision compares favorably to laboratory data provided that we ignore non-integrable effects which are of small amplitude relative to the solitons in small amplitude regime \cite{RNMOM19}.
Because of the finite size of the flume, experimental observations were necessarily of counter propagating solitons to achieve the length of propagation needed to produce the soliton gas \cite{RNMOM19}.
In this paper we consider the problem of computing solitons to the Kaup--Broer system consisting of many interacting solitons. 
This is done by considering the primitive solutions, which were first introduced for the KdV equation \cite{DZZ16,ZDZ16,ZZD16,NZZ19}.
We also discuss how the method appears as a reduction of the dressing method for the complete complexification of the 2+1D completely integrable generalization of the Kaup--Broer system originally considered by Rogers and Pashaev \cite{RP11}.

It should be noted that the Kaup--Broer system has been shown to be ill-posed in Sobelev space \cite{ABM17}.
However, this does not mean there can be no practical application of particular solutions to the Kaup--Broer system in the small amplitude and long wave regime.
This does mean trying to apply conventional numerical or perturbation theory methods to the Kaup--Broer system is not a good idea.
However, given that the collision of two counter propagating solitons modeled by the Kaup--Broer system compares favorably to experiment \cite{RNMOM19}, a soliton gas described by the Kaup--Broer system could be interesting from a practical point of view.
Moreover, as mentioned in \cite{RNMOM19}, the observations of Chen and Yeh \cite{CY14} that show that the head on collision of solitons produces an amplification over a simple linear superposition of KdV solitons means we can expect a counter propagating soliton gas to have larger amplitude statistics than a sum of counter propagating KdV soliton gases.
The splitting of the scattering theory into left and right moving components makes the Kaup--Broer system an intuitive model for describing a soliton gas consisting of counter propagating solitons.

In 1975 Kaup \cite{K75} and Broer \cite{B75} independently derived the following system of nonlinear partial differential equations
\begin{align} \label{eq:KaupSystem1dim}
& \eta_t + h_0 \varphi_{xx} + (\eta \varphi_x)_x + \left(\frac{h_0^3}{3} - \frac{h_0 \tau}{\rho g} \right) \varphi_{xxxx}  = 0, \\
& \varphi_t + \frac{1}{2} (\varphi_x)^2 + g \eta = 0,
\label{eq:KaupSystem2dim}
\end{align}
which we have expressed in dimensional form.
We will call this system of equations the Kaup--Broer system.
This system describes weakly nonlinear long shallow water waves in a channel of constant depth.
In system (\ref{eq:KaupSystem1dim},\ref{eq:KaupSystem2dim}) the constant $h_0$ is the quiescent water depth, $g$ is the gravitational acceleration, $\tau$ is the surface tension of the fluid, and $\rho$ is the density of the fluid. 
The variable $\eta$ is the free surface displacement from quiescent water depth, and $\varphi$ is the velocity potential evaluated on the free surface.

Broer arrived on the system (\ref{eq:KaupSystem1dim},\ref{eq:KaupSystem2dim}) by considering Hamiltonian approximations to the Hamiltonian equations for water--waves in a narrow channel.
Kaup showed that system (\ref{eq:KaupSystem1dim},\ref{eq:KaupSystem2dim}) when $g>0$ and $\tfrac{h_0^2}{3} > \tfrac{\tau}{g \rho}$ can be solved by the Inverse Scattering Method (ISM);
he studied the direct and inverse scattering problem, and found one-solition solutions.
Periodic and quasiperiodic finite gap solutions and their N-soliton limit were found by Matveev and Yavor \cite{MY79}.
The N-soliton solution have also been computed using the Hirota bilinear form \cite{W13}.
Complete integrability of the Kaup--Broer system was proven by Kupershmidt \cite{K86}.

Consider the general form of the Kaup--Broer system
\begin{equation} \label{eq:KaupSystemg1}
\eta_t + \mu_1 \varphi_{xx} + \mu_2(\eta \varphi_x)_x + \mu_3 \varphi_{xxxx} = 0,
\end{equation}
\begin{equation} \label{eq:KaupSystemg2}
\varphi_t + \varepsilon_1 \frac{1}{2} (\varphi_x)^2 + \varepsilon_2 \eta = 0.
\end{equation}
The system (\ref{eq:KaupSystemg1}, \ref{eq:KaupSystemg2}) has 4 distinct scaling classes that can be represented by systems of the form
\begin{equation} \label{eq:KaupSystem1}
\eta_t + \varphi_{xx} + (\eta \varphi_x)_x + \mu \varphi_{xxxx} = 0,
\end{equation}
\begin{equation} \label{eq:KaupSystem2}
\varphi_t + \frac{1}{2} (\varphi_x)^2 + \varepsilon \eta = 0,
\end{equation}
determined by the 4 possible choices defined by $\varepsilon = \pm 1$ and $\mu = \pm \tfrac{1}{4}$.
The 4 scaling class of the Kaup--Broer system are analogous to the 4 scaling classes of 1D Boussinesq equations discussed by Bogdanov and Zakharov in \cite{BZ02}.
The 4 scaling classes of the Kaup--Broer system have the following physical interpretations:
\begin{enumerate}
\item $\varepsilon = 1$, $\mu = \frac{1}{4}$ corresponds to gravitational non-capillary waves.

\item $\varepsilon = 1$, $\mu = -\frac{1}{4}$ corresponds to gravitational capillary waves.

\item $\varepsilon = -1$, $\mu = \frac{1}{4}$ corresponds to non-capillary waves in reversed gravity.

\item $\varepsilon = -1$, $\mu = -\frac{1}{4}$ corresponds to capillary waves in reversed gravity.
\end{enumerate}

It is known (see, for instance \cite{ZM84}) that any integrable system in 1+1 dimensions admits a generalization to an equation in 2+1 dimensions, preserving integrability.
For instance, the KdV equation
\begin{equation} \label{eq:KdV}
u_t + 6 u u_x + u_{xxx} = 0
\end{equation}
admits generalizations to the KP equation
\begin{equation} \label{eq:KP}
\frac{\partial}{\partial x} (u_t + 6 u u_x + u_{xxx}) = \pm 3 u_{yy},
\end{equation}  
and the Veselov--Novikov equation
\begin{align}
& u_t = u_{zzz} + u_{\bar z \bar z \bar z} + (uv)_z + (u \bar v)_{\bar z} \\
& v_{\bar z} = -3u_z, u = \bar u. 
\end{align}
A completely integrable 2+1 dimensional generalization of the Kaup--Broer system was discovered by Rogers and Pashaev using the Hirota bilinear form \cite{RP11}.

All four scaling classes of the Kaup--Broer system are dimensional reductions of the completely complexified version of the system
\begin{align}
\label{eq:ssys} &s_t + \alpha s_u^2-\beta s_{v}^2 + \Pi = 0 \\
\label{eq:csys} & c_t+2\alpha (s_u c)_u - 2\beta(s_v c)_v  - 2 a P s + \frac{1}{2} P s_{uv}
\end{align}
with
\begin{equation} \label{eq:Pdef} P = \alpha \frac{\partial^2}{\partial u^2} - \beta \frac{\partial^2}{\partial v^2} \end{equation}
and $\Pi$ is defined in terms of $c$ by solving
\begin{equation} \label{eq:Pisys} \Pi_{uv} = 2Pc. \end{equation}
By complete complexification, we mean all coordinates, fields, and coefficients are taken to possibly be complex numbers. 
The $\Pi$ dependence in (\ref{eq:ssys}) can be removed by differentiating (\ref{eq:ssys}) by $u$ and $v$ and applying (\ref{eq:Pisys}). 
This completely integrable generalization of the Kaup--Broer system was first introduced by Rogers and Pashev \cite{RP11} (for a particular choice of parameters).
The system (\ref{eq:ssys},\ref{eq:csys}) is Hamiltonian in the sense that it is the the canonical system
\begin{equation}
s_t = -\frac{\delta H}{\delta c}, \quad c_t = \frac{\delta H}{\delta s}
\end{equation}
for the Hamiltonian functional
\begin{align}
 H & = \iint \alpha cs_u^2 - \beta c s_v^2 + asPs -\frac{1}{4} s Ps_{uv} + \frac{1}{2} c \Pi \: du dv.
\end{align}

We can see that all 4 scaling classes of the Kaup--Broer system are achievable as the dimensional reduction of the complete complexification of integrable 2+1 dimensional generalization of the Kaup--Broer system.
Consider the choice parameters $\alpha = -\beta = \frac{1}{4}$ and $a = \pm 1$, suppose there exists a solution $s,c$ to (\ref{eq:ssys},\ref{eq:csys}) such that 
\begin{equation} \label{eq:reduction} \sigma s(u,v, \sigma t) = \varphi(u+v,t), \quad -ac(u,v,\sigma t) = \eta(u+v,t). \end{equation}
If we write $x = u+v$, then $P$ acts on the 1+1 dimensional fields as $\frac{1}{2}\frac{\partial^2}{\partial x^2}$, and the equation for $\Pi$ becomes $\Pi_{xx} = -a\eta_{xx}$.
It is thus easy to see that $\varphi(x,t)$ and $\eta(x,t)$ solve
\begin{align}
& \varphi_t + \frac{1}{2} \varphi_x^2 - a \sigma^2 \eta = 0 \\
& \eta_t + (\varphi_x \eta)_x  + \varphi_{xx} - \frac{a}{4} \varphi_{xxxx} = 0.
\end{align}
We see that all 4 scaling classes are attainable from the complete complexification based on the two choices $a = \pm 1$ and $\sigma = 1,i$.
The 2+1 dimensional Hamiltonian then reduces to
\begin{align}
 H & = \int \frac{1}{2} \eta \varphi_x^2 - \frac{1}{2} \varphi \varphi_{xx} -\frac{a}{4} \varphi \varphi_{xxxx} - \frac{a \sigma^2}{2} \eta^2 \: dx.
\end{align}

\begin{rmk}
The fact that we can solve the complete complexification allows us to produce solutions to other interesting reductions of the completely complexified equations.
A particularly interesting reduction from the complete complexificaiton corresponds to taking  $u = z = x+iy$ and $v = \bar z = x-iy$ with $\alpha = \beta = 1$ so that
\begin{equation} \frac{\partial^2}{\partial u \partial v} = \frac{1}{4} \left(\frac{\partial^2}{\partial x^2} + \frac{\partial^2}{\partial y^2}\right) = \frac{1}{4} \Delta, \end{equation}
\begin{equation} P = \frac{1}{2} \left( \frac{\partial^2}{\partial x^2} - \frac{\partial^2}{\partial y^2} \right) = \frac{1}{2} \Box, \end{equation}
and the 2+1 dimensional equation becomes
\begin{align} & s_t - \frac{i}{2} s_x s_y + \Pi = 0, \\
& c_t -i (s_xc)_y - i (s_yc)_x - a\Box s + \frac{1}{16} \Delta \Box s = 0, \\
& \Delta \Pi = 4 \Box c.  \end{align}
The linearization of this system can be expressed in terms of $c$ and $\Pi$ as
\begin{align}
& \label{eq:linearizationcrmk} c_{tt} + a \Box \Pi-\frac{1}{16} \Delta \Box \Pi = 0, \\
& \label{eq:linearizationPirmk} \Delta \Pi = 4 \Box c.
\end{align}
If we take the Fourier transform
\begin{equation} c_{pq} = \frac{1}{2\pi} \iint c(x,y) e^{-i(px+qy)} \: dx dy, \end{equation}
then equation (\ref{eq:linearizationcrmk}) becomes
\begin{equation} (c_{pq})_{tt} + a (p^2-q^2)\Pi_{pq} - \frac{1}{16} (p^2+q^2)(p^2-q^2) \Pi_{pq} = 0.  \end{equation}
Moreover, (\ref{eq:linearizationPirmk}) implies
\begin{equation} \Pi_{pq} = 4 \frac{p^2-q^2}{p^2+q^2} c_{pq}. \end{equation}
so we can remove $\Pi_{pq}$ from the equation for $c_{pq}$ to get
\begin{equation} (c_{pq})_{tt} + 4 a \frac{(p^2-q^2)^2}{p^2+q^2} c_{pq} - \frac{1}{4} (p^2-q^2)^2 c_{pq} = 0. \end{equation}
Therefore, this version of the 2+1 dimensional system is a nonlinear wave equation with cubic nonlinearity, and dispersion relation
\begin{equation} \omega^2 = \frac{1}{4} (p^2-q^2)^2 - 4 a \frac{(p^2-q^2)^2}{p^2+q^2}. \end{equation}
When $a<0$ this is a stable wave equation in the long wave limit, while for $a > 0$ this is stable only in the long wave limit with $p >> q$ or $q >> p$.
These could potentially be physically relevant wave equations.
\end{rmk}

In this paper we will construct solutions to system (\ref{eq:KaupSystem1},\ref{eq:KaupSystem2}) and its generalization (\ref{eq:ssys},\ref{eq:csys}) by the dressing method developed in papers of Zakharov and his collaborators \cite{Z89,ZM84,ZM85,BM88}.
The dressing method makes it possible to find much broader classes of solutions of integrable systems than the ISM, and is based on the use of the nonlocal $\bar{\partial}$-problem first used by Ablowits, Fokas, and Bar-Yaacov \cite{ABF83} for the solution of the KP-2 equation (the KP-2 is equation (\ref{eq:KP}) with the positive choice of sign).
Moreover, since $u,v,t$ and $\alpha,\beta,a$ will appear as parameters in the Kaup--Broer system, this method leads immediately to solutions of the completely complexified system as well.

One important property of the 2+1 dimensional generalization of the Kaup--Broer system is that the nonlocal $\bar \partial$ problem for the 4 cases of the Kaup--Broer system are dimensional reductions of the nonlocal $\bar \partial$ problem for the 2+1 dimensional generalization.
The operator for the scattering problem for our 2+1 dimensional generalization of the Kaup--Broer system is the non-relativistic quantum Hamiltonian for a charged particle under the influence of a magnetic field.

In the case of $\varepsilon = 1$ and $\mu = \tfrac{1}{4}$, we will explicitly construct the N-soliton solutions, and also construct solutions that can be interpreted as a limit of the N-soliton solutions as the number of solitons diverges to $\infty$.
We will provide numerical evidence that the second type of solution can describe counter propagating dispersive shockwave type solutions.

The problem of computing solutions to completely integrable partial differential equations that are limits of sequences N-soliton solutions as the number of solitons diverges to infinity is an intriguing area of research with many interesting unsolved problems.
Recent progress on this problem was made by Dyachenko, Zakharov and Zakharov who demonstrated in \cite{DZZ16,ZDZ16,ZZD16} that periodic, dispersive shockwave, and turbulent solutions to the KdV equation can be computed as limits of N-soliton solutions.
From the point of view of the inverse spectral theory of 1D Schr\"{o}dinger operators, a continuum limit was taken in which discrete eigenvalue coalesce into spectral bands.
These potentials were dubbed primitive potentials.
This work was continued in \cite{NZZ19} where the case of symmetric primitive potentials was considered in detail.
An analytic procedure for computing all the Taylor coefficients of symmetric primitive potentials was presented, as well as the special case of an elliptic potential.
The asymptotic behavior of genus one dispersive shockwave behavior of a soliton was analyzed rigorously by Girotti, Grave and McLaughlin \cite{GGM18} using the Riemann--Hilbert problem formulation of the infinite soliton limit presented by Dyachenko, Zakharov and Zakharov \cite{DZZ16,ZDZ16,ZZD16}.
It was shown how to produce all finite gap solutions using primitive solution in \cite{N19}.

 One draw back of using the KdV equation as a model of 1+1D (one spatial and one temporal dimension) shallow water long waves is that the KdV equation is derived by assuming wave motion in a single direction.
A common completely integrable alternative to the KdV equation is the 1+1 dimensional Boussinesq equation.
However, one issue is that the scattering problem used to solve the Boussinesq equation by the inverse scattering transform (IST) is of degree three making application of nonlinear steepest descent more difficult \cite{BZ02}.
However, the Kaup--Broer system which has a degree 2 scattering problem that separates into left and right moving KdV scattering problems.

A more in depth comparison of the dressing method for the Kaup--Broer system to the dressing method for the Korteweg--de Vreis and the Kadomtsev--Petviashvili equations can be found in chapter 2 of the first author's PhD dissertation \cite{N18} done under the supervision of the second author.
However, in the PhD dissertation only N-soliton solutions were computed by the dressing method.
In this paper we extend the results first presented in the dissertation to compute solutions to the Kaup--Broer system in the closure of the N-soliton solutions with respect to the topology of uniform convergence in compact sets.

\subsection{Outline of Paper}

This paper is structured as follows:
\begin{itemize}
\item In section 2 we provide the Manakov triple formulation of the 2+1 dimensional generalization of the Kaup---Broer system.
\item In section 3 we give the nonlocal $\bar \partial$ formulation of the dressing method for the 2+1 dimensional generalization of Kaup--Broer system, and use it to formally compute a new class of solutions to this generalization in subsection 3.1.
We then discuss how the dimensional reduction of the 2+1 dimensional system to the 1+1 dimensional Kaup--Broer system is achieved from the point of view of the nonlocal $\bar \partial$ problem in subsection 3.2.
\item In section 4 we discuss the most studied case of the 1+1 dimensional Kaup--Borer system ($\epsilon = 1, \mu = \frac{1}{4}$), and show how the dressing method can be used to produce a new class of solutions to the 1+1 dimensional Kaup--Broer system called primitive solutions that can be interpreted as the infinite soliton limit of the N-soliton solutions. 
\item In section 5 we discuss how finite gap solutions to the Kaup--Broer system can be computed 
\item In section 6 we compute numerical approximations (corresponding to exact N-soliton solutions with large N) to some of these new solutions to the 1+1 Kaup--Broer system.
\item In section 7 we provide some concluding remarks.
\end{itemize}

\section{The Manakov Triple for the 2+1 Dimensional Generalization of the Kaup--Broer System}

We will construct the Manakov triple and the corresponding linear system for the 2+1 dimensional generalization of the Kaup--Broer system.
The existence of the Manakov triple is what justifies the use of the dressing method discussed in the next section. 
The system (\ref{eq:ssys},\ref{eq:csys}) is equivalent to a solution of the Manakov triple operator equation
\begin{equation}  \label{eq:Manakov}
L_t = [M,L] + Q L,
\end{equation}
where
\begin{equation} L = \frac{\partial^2}{\partial u \partial v} + A \frac{\partial}{\partial v} + B - a, \end{equation}
\begin{equation} M = \alpha \frac{\partial^2}{\partial u^2} + \beta \frac{\partial^2}{\partial v^2} + F \frac{\partial}{\partial v} + G, \end{equation}
and $Q = F_v - 2 \alpha A_u$.
The equation for the Manakov triple is equivalent to the simultaneous solvability of the linear system
\begin{equation} \label{eq:linsysMan} L \psi = 0, \quad \psi_t = M \psi. \end{equation}
In the next section we will discuss how the nonlocal $\bar \partial$ problem can be used to compute a solution $\psi$ to (\ref{eq:linsysMan}), corresponding potentials, and a solution to the 2+1 dimensional generalization of the Kaup--Broer system.

In terms of $A,B,F$ and $G$, the Manakov triple equation is
\begin{align} A_t \frac{\partial}{\partial v} + B_t = & (2 \beta A_v - F_u) \frac{\partial^2}{\partial u^2} 
 + (2 \alpha B_u - G_v) \frac{\partial}{\partial u} \\ & + (\alpha A_{uu} + \beta A_{vv}  + 2 \beta B_v - F_{uv} + FA_v - F_vA - G_u) \frac{\partial}{\partial v} \\ & + \alpha B_{uu} + \beta B_{vv} + F B_v - G_{uv} - A G_v - a F_v + 2 \alpha a A_u .\end{align}
The coefficients for the constant, $\frac{\partial}{\partial v}$, $\frac{\partial}{\partial u}$, $\frac{\partial^2}{\partial u^2}$ terms solve the nonlinear system
\begin{align} & \label{eq:nonlinA} A_t = \alpha A_{uu} + \beta A_{vv}  +  2 \alpha B_u + 2 \beta B_v - F_{uv} + FA_v - F_vA - G_u \\
& \label{eq:nonlinB} B_t = \alpha B_{uu} + \beta B_{vv} + F B_v - G_{uv} - A G_v - a F_v + 2 \alpha a A_u \\
& \label{eq:linF} F_u = 2 \beta A_v, \\ 
& \label{eq:linG} G_v = 2 \alpha B_u.
 \end{align}
We will now show that this system is equivalent to the complete complexification of the 2+1 dimensional completely integrable generalization of the Kaup--Broer system.
If we suppose the $A$ has the potential form $A = s_u$ for some potential $s$ (we will this potential exists when we discuss the dressing method in the next section) then $F = 2 \beta s_v + C$ for some constant $C$.
Since we will remove the $F$ dependence we can take $C = 0$ and produce a solution to the nonlinear system (\ref{eq:nonlinA}-\ref{eq:linG}).
There also exist some antiderivative $\partial_v^{-1}$ such that  $G = 2 \alpha \partial_v^{-1} B_u$.



Eliminating the $F$ and $G$ dependance in (\ref{eq:nonlinA},\ref{eq:nonlinB}), replacing $A$ with the potential $s$ and anti-differentiating equation (\ref{eq:nonlinA}) with respect to $u$ and $v$ implies
\begin{align}
& \label{eq:ssys0} s_t + \alpha (s_u)^2 - \beta (s_v)^2 - Ps + \rho = 0, \\
& \label{eq:Bsys} B_t + \alpha  2 (s_u B)_u  - \beta 2 (s_v B)_v + PB - a P s   = 0,
\end{align}
where $\rho$ is defined in terms of $B$ as a solution to
\begin{equation}
\label{eq:rhosys} \rho_{uv} = 2 PB, \end{equation}
and we recall that $P$ is the differential operator operator
\begin{equation*} P = \alpha \frac{\partial^2}{\partial u^2} - \beta \frac{\partial^2}{\partial v^2}. \end{equation*}

The $\rho$ dependence in the system can be removed by differentiating (\ref{eq:ssys0}) by $u$ and $v$ and applying (\ref{eq:rhosys}).
We can change variables to fields $c$ and $s$ related by
\begin{equation} c = B - \frac{1}{2} s_{uv} \end{equation}
so that we end up with the system
\begin{align*}
&s_t + \alpha s_u^2-\beta s_{v}^2 + \Pi = 0 \\
& c_t+2\alpha (s_u c)_u - 2\beta(s_v c)_v  - 2 a P s + \frac{1}{2} P s_{uv} = 0
\end{align*}
with $\Pi$ defined in terms of $c$ as the solution to
\begin{equation*} \Pi_{uv} = 2Pc.\end{equation*}
This is the 2+1 dimensional generalization of the Kaup--Broer system.
The function $\Pi$ so defined is related to $s$ and $\rho$ by $\Pi = \rho-Ps$.
This is the system discussed in the introduction.
Moreover, this change of variables is a canonical transformation.

\begin{rmk}
If we transfer to complex coordinates $u=x+iy$ and $v=x-iy$, then operator $L$ is
\begin{equation}
L = \frac{1}{4} \left(\frac{\partial^2}{\partial x^2} + \frac{\partial^2}{\partial y^2}\right) + \frac{A}{2} \left(\frac{\partial}{\partial x} + i \frac{\partial}{\partial y}\right) + B - a,
\end{equation}
which has the physical interpretation as being proportional to the quantum mechanical Hamiltonian for a 2 dimensional non-relativistic charged quantum particle under the influence of an electromagnetic field.
The solutions to the nonlocal $\bar \partial$ problem will allow us to compute solutions to $L\psi=0$.
The inverse spectral theory of periodic and finite gap operators of this form was studied extensively by Novikov and his collaborators, see \cite{D81,N83} for a review of this theory.

\end{rmk}

\section{The Dressing Method for the 1+1 Dimensional Kaup--Broer System and Its 2+1 Dimensional Generalization}

The nonlocal $\bar \partial$ problem can be used to find solutions $\psi$ to the system (\ref{eq:linsysMan}) with coefficients in the linear operators depending on the solution.
Since the coefficients depend on $\psi$, this method computes  a solution to the 2+1 dimensional generalization of the Kaup--Broer system using $\psi$.
We can produce a meromorphic family of solutions $\psi(\lambda;u,v,t)$ to (\ref{eq:linsysMan}) by taking $\psi(\lambda;u,v,t) = e^{\phi(\lambda;u,v,t)} \chi(\lambda;u,v,t)$ where
\begin{equation}
\phi(\lambda;u,v,t) = \lambda u + a \lambda^{-1} v + (\alpha \lambda^2 + \beta \lambda^{-2}) t
\end{equation} 
and $\chi$ solves the nonlocal $\bar \partial$ problem
\begin{equation} \label{eq:dbar2p1} \frac{\partial \chi}{\partial \bar{\lambda}} (\lambda;u,v,t) = \iint_{\mathbb{C}} R(\lambda,w) e^{\phi(w;u,v,t) - \phi(\lambda;u,v,t)} \chi(w;u,v,t) dA(w)  \end{equation}
and is normalized so $\chi \to 1$ as $\lambda \to \infty$.
Equivalently, $\chi$ solves the equivalent integral equation
\begin{equation} \chi(\lambda;u,v,t) = 1 + \frac{1}{\pi} \iint_{\mathbb{C}} \iint_{\mathbb{C}} \frac{e^{\phi(w;u,v,t) - \phi(\zeta;u,v,t)}}{\lambda - \zeta} R(\zeta,w) \chi(w;u,v,t) dA(w) d A(\zeta).  \end{equation}
Define the operators
\begin{align} \label{eq:L1s}
L_1 \chi = & D_u D_v \chi + A D_v \chi + (B - a) \chi \\ \label{eq:L1e}
= & \chi_{uv} + \lambda \chi_v + a \lambda^{-1} \chi_u + A \chi_v + A a \lambda^{-1} \chi + B \chi,
\end{align}
and
\begin{align} \label{eq:L2s}
L_2 \chi = & D_t \chi - \alpha D_u^2\chi - \beta D_v^2 \chi - F D_v \chi - G \chi \\ \label{eq:L2e}
= & \chi_t - \alpha (\chi_{uu} + 2 \lambda \chi_u) - \beta\left(\chi_{vv} + 2 a \lambda^{-1} \chi_v \right) - F \left(\chi_v + a \lambda^{-1} \chi \right) - G \chi,
\end{align}
where
\begin{equation} D_u = \frac{\partial}{\partial u} + \lambda, \quad D_v = \frac{\partial}{\partial v} + a \lambda^{-1}, \quad D_t = \frac{\partial}{\partial t} + (\alpha \lambda^2 + \beta \lambda^{-2}) t.  \end{equation}
Applying $L_j$ to the nonlocal $\bar \partial$ problem (\ref{eq:dbar2p1}) we find that $L_j \chi$ solve the nonlocal $\bar \partial$ problem
\begin{equation} 
\frac{\partial L_j \chi}{\partial \bar{\lambda}} (\lambda) = \gamma_j \delta(\lambda) + \iint_{\mathbb{C}} R(\lambda, w) e^{\phi(w;u,v,t) - \phi(\lambda;u,v,t)} L_j \chi(w;u,v,t) dA(w), \end{equation}
for $L_j \chi$, where
\begin{equation}
\gamma_1 = \pi a \left( \chi_u(0) + A \chi(0) \right), \gamma_2 = -\pi a(2 \beta \chi_v(0) + F \chi(0))
\end{equation}
are such that
\begin{equation} \left[\frac{\partial}{\partial \bar \lambda},L_j\right] = \gamma_j \delta(\lambda). \end{equation}
The condition that $\gamma_j = 0$ is equivalent to the assumption that $L_j \chi(\lambda)$ is holomorphic at $\lambda = 0$. 

The crux is that if $A,B,F$ and $G$ are set according to
\begin{equation} \label{eq:potentials} s = -\log(\chi_0^0), \quad A = s_u, \quad B = - (\chi_1^\infty)_{v}, \quad F = 2\beta s_v, \quad G = -2\alpha (\chi_1^\infty)_{u},  \end{equation}
then $L_j \chi \to 0$ as $\lambda \to \infty$, and $\gamma_j = 0$ making $L_j \chi$ holomorphic off of $\lambda \in \mathbb{C} \setminus \{\text{poles}\}$.
If it can be proved that the integral equation for $\chi$ is a Fredholm integral equations of index $0$, then
\begin{equation} \label{eq:Manlindress} L_1 \chi(\lambda; u,v,t) = 0, \quad L_2 \chi(\lambda; u,v,t) = 0 \end{equation}
because $L_j \chi$ must satisfy homogenous Fredholm integral equations of index $0$.
It is easy to see that $\chi$ solves (\ref{eq:Manlindress}) if and only if $e^{\phi(\lambda;u,v,t)} \chi(\lambda;u,v,t)$ solves (\ref{eq:linsysMan}).
Since we have produced a meromorphic family of solutions to (\ref{eq:linsysMan}), the functions $s$ and $c = B-\frac{1}{2} s_{uv}$ solve the 2+1 dimensional Kaup--Broer system.
In the soliton case, the integral equations is clearly Fredholm because it is equivalent to a finite dimensional linear equation. 

\begin{thm}
Suppose the $R : \mathbb{C}^2 \to \mathbb{C}$ is a generalized function that has been chosen so that the operator $\mathcal{F}$ defined by
\begin{equation} \mathcal{F} \chi(\lambda;u,v,t) = \frac{1}{\pi} \iint_{\mathbb{C}} \iint_{\mathbb{C}} \frac{e^{\phi(w;u,v,t) - \phi(\zeta;u,v,t)}}{\lambda - \zeta} R(\zeta,w) \chi(w;u,v,t) dA(w) d A(\zeta)\end{equation}
is a Fredholm operator of index $0$ with a trivial null space.
Then there is a unique solution to the integral equation
\begin{equation}  (\mathcal{I} - \mathcal{F}) \chi(\lambda;u,v,t) = 1. \end{equation}
The function $\chi$ solves the $\bar \partial$ problem
\begin{equation} \frac{\partial \chi}{\partial \bar{\lambda}}(\lambda;u,v,t) = \iint_{\mathbb{C}} R(\lambda,w) e^{\phi(w;u,v,t) - \phi(\lambda;u,v,t)} \chi\left(w;u,v,t\right) dA(w), \end{equation}
and is the unique solution normalized by $\chi(\lambda) \to 1$ as $\lambda \to \infty$.

The function $\chi$ has expansions
\begin{equation} \chi(\lambda;u,v,t) = \sum_{n = 0}^\infty \chi_n^0(u,v,t) \lambda^n, \quad \chi(\lambda;x,t) = \sum_{n = 0}^\infty \chi_n^\infty (u,v,t) \lambda^{-n}.   \end{equation}
The functions $c$ and $s$ computed from
\begin{equation} \label{eq:sol2p1} s(u,v,t) = - \log(\chi_0^0(u,v,t)), \quad  c(u,v,t) = -(\chi_{1}^\infty)_v (u,v,t) - \frac{1}{2} s_{uv} (u,v,t),     \end{equation}
and $\Pi$ determined from $c$ by solving the linear differential equation
\begin{equation} \Pi_{uv} = 2Pc \end{equation}
solve the complete complexification of the integrable 2+1 dimensional generalization of the Kaup--Broer system (\ref{eq:ssys},\ref{eq:csys}).
\end{thm}

\begin{rmk}
In principle, all systems obtained from the solution $\chi$ to the nonlocal $\bar \partial$ problem (\ref{eq:dbar2p1}) are conservation laws of (\ref{eq:ssys},\ref{eq:csys}).
The function $\chi$ is nothing but the famous $\tau$-function used in the Sato theory.
This is a consequence of the theory discussed in \cite{ZM84,ZM85,Z89}.
\end{rmk}

\begin{rmk}
The operators $D_{u,v,t}$ and $L_{1,2}$ and the dressing functions $R$ with $a = 1$ discussed in this section are closely related to those used to solve the Veselov--Novikov equation with the dressing method \cite{B87}.
The differences are:
\begin{itemize}
\item We replace the ``light cone" coordinates $u$, $v$ with complex conjugate coordinates $u \to \bar z$, and $v \to z$.
\item We replace $\alpha \lambda^2 + \beta \lambda^{-2}$ in the definitions of $D_t$ and $\phi$ with $\lambda^3 + \lambda^{-3}$.
\end{itemize}
The Riemann--Hilbert problem used in \cite{B87} is really just a nonlocal $\bar \partial$ problem in the case where $R$ is supported on an arc.
By setting $A = 0$ and replacing $u$ and $v$ with the complex conjugated coordinates $\bar z$ and $z$, we end up with the operator $L$ evolving in the lax equation for the Veselov--Novikov equation \cite{B87,M76,NV84}.
\end{rmk}

\subsection{A Class of Exact Solutions to the 2+1 Dimensional Kaup--Broer System Generalizing the N-Soliton Solutions}
\label{ssec:Nsolgen}
We can consider a completely solvable case of the nonlocal $\bar \partial$ problem complex 2+1 dimensional system corresponding to a dressing function of the form
\begin{equation} R(\lambda,w) = \sum_{n = 1}^N f_n(\lambda) g_n(w). \end{equation}
Then the solution to the nonlocal $\bar \partial$ problem is of the form
\begin{equation} \label{eq:chigensol}
\chi(\lambda;u,v,t) = 1 + \frac{1}{\pi} \sum_{n = 1}^N \Phi_n(u,v,t) \hat f_n(u,v,t) ,
\end{equation}
where $\hat f_n$ denotes the integral transform 
\begin{equation} \hat f_n (\lambda;u,v,t) = \iint_{\mathbb{C}} \frac{f_n(\zeta) e^{-\phi(\zeta;u,v,t)}}{\lambda - \zeta} dA(\zeta) \end{equation}
of $f_n$, and
\begin{equation} \label{eq:chigenphi}
\Phi_n(u,v,t) = \iint_{\mathbb{C}} g_n(w) e^{\phi(w;u,v,t)} \chi(w;u,v,t) dA(w).
\end{equation}
Substitution of (\ref{eq:chigensol}) into the definition of $\Phi_n$ gives the equation
\begin{equation}
\Phi_n(u,v,t) = G_n(u,v,t) + \sum_{m = 1}^N M_{nm}(u,v,t) \Phi_m(u,v,t),
\end{equation}
where
\begin{equation} G_n(u,v,t) = \iint_{\mathbb{C}} g_n(w) e^{\phi(w;u,v,t)} dA(w), \end{equation}
and
\begin{equation} \small M_{nm}(u,v,t) = \frac{1}{\pi} \iint_{\mathbb{C}} \iint_{\mathbb{C}} \frac{f_m(\zeta) g_n(w) e^{\phi(w;u,v,t) - \phi(\zeta;u,v,t)}}{w - \zeta} dA(\zeta) dA(w). \end{equation}
If $\det(I-M) \ne 0$ then the functions $\chi_0^0(u,v,t)$ and $\chi_1^\infty(u,v,t)$ given by
\begin{equation} \chi_1^\infty(u,v,t) = \frac{1}{\pi} \sum_{n = 1}^N \Phi_n(u,v,t) F_n(u,v,t) \end{equation}
where
\begin{equation} F_n(u,v,t) = \iint_{\mathbb{C}} f_n(\zeta) e^{-\phi(\zeta;u,v,t)} dA(\zeta), \end{equation}
and
\begin{equation} \chi_0^0(u,v,t) = 1 + \hat f_n(0;u,v,t),  \end{equation}
can then be used to compute a solution $s$, $c$ to the 2+1 dimensional Kaup--Broer system by theorem 1. When $f_n$ and $g_n$ are taken to be constant multiples of delta functions supported on some points $z_n$, $w_n$, this solution reduces to an $N$-soliton solution.

\subsection{Reduction of the 2+1 Dimensional Dressing Method to 1+1 Dimensions}

We now present the Dressing Method for the Kaup--Broer system.
The dressing method allows solutions to all 4 scaling classes of the Kaup--Broer to be computed by solving a nonlocal $\bar \partial$ problem.

From (\ref{eq:reduction}) we see that dimensional reduction to the 4 scaling classes of the Kaup--Broer system can be achieved provided the solution depends only on $x = u+v$, and $\alpha = -\beta = \frac{1}{4}$.
The $u$, $v$, and $t$ dependence enters the $\bar \partial$ problem only in terms of the combination
\begin{equation} \phi(w;u,v,t) - \phi(\lambda;u,v,t) = -(\lambda - w) u - a (\lambda^{-1} - w^{-1}) v - (\alpha \lambda^2 + \beta \lambda^{-2} - \alpha w^2 - \beta w^{-2}) t.\end{equation}
This will depend only on $x = u+v$ if and only if
\begin{equation} \lambda - w = a (\lambda^{-1} - w^{-1}) \end{equation}
which is solved if $w$ relates to $a$ by $w = -a\lambda^{-1}$.
Therefore suppose that
\begin{equation} R(\lambda,w) = \delta(w+a\lambda^{-1}) R_0(\lambda)\end{equation} where $R_0$ has compact support on $\mathbb{C} \setminus \{-1,0,1\}$.
Then (\ref{eq:dbar2p1}) reduces to the nonlocal $\bar \partial$ problem 
\begin{equation} \label{eq:dbar1p1}
\frac{\partial \chi}{\partial \bar \lambda}(\lambda;x,t) =  R_0 (\lambda) e^{- 2\tilde\phi\left(\lambda;x,t\right)}\chi\left(- a \lambda^{-1}; x,t  \right),
\end{equation}
where
\begin{equation}
\tilde \phi(\lambda;x,t) = \frac{1}{2}  \left(\lambda + a \lambda^{-1} \right) x + \frac{1}{4} \left(\lambda^2 - \lambda^{-2} \right) t.
\end{equation}
This nonlocal $\bar\partial$ problem depends only on $x$, so the solution to (the 2+1 dimensional generalization of the).
We look for a solution $\chi$ to the nonlocal $\bar \partial$ problem normalized so that $\chi(\lambda) \to 1$ as $\lambda \to \infty$.
This solution $\chi$ solves the integral equation
\begin{equation} \chi(\lambda;x,t) = 1 + \mathcal{F}_0 \chi(\lambda) \iff (\mathcal{I} - \mathcal{F}_0) \chi = 1 , \end{equation}
where $\mathcal{I}$ is the identity operator and $\mathcal{F}_0$ is the integral operator given by
\begin{equation} \label{eq:defFop} \mathcal{F}_0 \chi(\lambda;x,t) =  \frac{1}{\pi} \iint_{\mathbb{C}} \frac{R_0(\zeta) e^{-2 \tilde \phi(\zeta;x,t)}}{\lambda - \zeta} \chi(- a\zeta^{-1};x,t) dA(\zeta). \end{equation}

Therefore the following is a corollary to theorem 1:
\begin{cor}
Suppose the dressing function $R_0 : \mathbb{C} \to \mathbb{R}$ has been chosen so that the operator $\mathcal{F}_0$ given by (\ref{eq:defFop}) is a Fredholm operator of index $0$ with a trivial null space.
Then there is a unique solution to the integral equation
\begin{equation} (\mathcal{I} - \mathcal{F}_0) \chi(\lambda;x,t) = 1. \end{equation}
The function $\chi$ solves the nonlocal $\bar \partial$ problem
\begin{equation} \frac{\partial \chi}{\partial \bar \lambda}(\lambda;x,t) = R_0(\lambda) e^{-2 \tilde \phi(\lambda;x,t)} \chi\left(- a\lambda^{-1};x,t\right), \end{equation}
and is the unique solution normalized by $\chi(\lambda) \to 1$ as $\lambda \to \infty$.

The function $\chi$ has expansions
\begin{equation} \label{eq:coeffs1p1} \chi(\lambda;x,t) = \sum_{n = 0}^\infty \chi_n^0(x,t) \lambda^n, \quad \chi(\lambda;x,t) = 1 + \sum_{n = 1}^\infty \chi_n^\infty (x,t) \lambda^{-n}.   \end{equation}
Then the functions $\varphi$ and $\eta$ defined by
\begin{equation}  \varphi(x,t) = - \sigma \log(\chi_0^0(x,\sigma t)), \quad  \eta(u,v,t) = a(\chi_{1}^\infty)_x (x,\sigma t) + \frac{1}{2} a \log(\chi_0^0(x,\sigma t)),     \end{equation}
solve the Kaup--Broer system (\ref{eq:KaupSystem1}, \ref{eq:KaupSystem2}) with $\varepsilon = -a\sigma^2$ and $\mu =- \frac{a}{4}$.
\end{cor}

\begin{rmk} Remark 3 should still applies. Therefore $\chi$ should be equivalent to the $\tau$ function in the Sato theory, and any system derivable from $\chi$ should be a conservation law of (\ref{eq:KaupSystem1}, \ref{eq:KaupSystem2}). In particular, the normal and reversed gravity systems are conservation laws for each other (although they will not both be real and nonsingular). \end{rmk}

\begin{rmk} \label{rmk:linsysred} The function $\psi(\lambda;x,t) = e^{\tilde \phi(\lambda;x,t)} \chi(\lambda;x,t)$ solves the linear system
\begin{equation} \tilde L(\lambda,t) \psi(\lambda;x,t) = 0,
\end{equation}
\begin{equation} \tilde M(\lambda,t) \psi(\lambda;x,t) =\psi_t(\lambda;x,t),\end{equation}
where
\begin{equation} \tilde L(\lambda,t) =  \frac{\partial^2}{\partial x^2} + A \frac{\partial}{\partial x} + B - \frac{\lambda-a\lambda^{-1}}{2} A - \left( \frac{\lambda+a\lambda^{-1}}{2} \right)^2,\end{equation}
\begin{equation} \tilde M(\lambda,t) =  \left(\frac{\lambda-a\lambda^{-1}}{2} + F\right) \frac{\partial}{\partial x} + G - \frac{\lambda-a\lambda^{-1}}{2} F .\end{equation}
The compatibility condition to be able to find a simultaneous solution to both equations in the system is equivalent to the Kaup--Broer system.
The potentials are set according to the (\ref{eq:potentials}) except we use coefficients from (\ref{eq:coeffs1p1}) and the $u$ and $v$ derivatives reduce to $x$ derivatives. \end{rmk}

Any choice of $R_0$ leads to a solution to the Kaup--Broer system, however most choices of $R_0$ will lead to complex solutions with singularities.
Reality conditions on $R_0$ must be imposed to guarantee the solutions so produced are real and bounded.

\section{N-Soliton Solutions and the Limit $N \to \infty$ for the 1+1 Dimensional Kaup--Broer System}

In this section we will consider the case of the N-soliton solutions in the 1+1 dimensional case with $a = -1$, $\sigma=1$ (the most studies case).
We then take the infinite soliton limit using the primitive solution method.

To produce an N-soliton solutions consider some parameters $\{ \lambda_n,r_n \}_{n = 1}^N$ which we will call the discrete spectral data for the N-soliton solutions.
We will assume $\lambda_n \in \mathbb{R}\setminus\{-1,0,1\}$, $r_n > 0$ if $\lambda_n \in (-1,0)\cup(1,\infty)$, and $r_n<0$ if $\lambda_n \in (-\infty,-1) \cup (0,1)$.
The positive $\lambda_n$ correspond to left moving solitons, while the negative $\lambda_n$ correspond to right moving solitons.

We take the dressing function \begin{equation} R_0(\lambda) = \sum_{n = 1}^N r_n \delta(\lambda - \lambda_n) \end{equation} so the integral equation becomes the algebraic equation
\begin{equation} \label{eq:Nsol}
\chi(\lambda;x,t) = 1 + \frac{1}{\pi} \sum_{n =1}^N \frac{r_n e^{-\tilde \phi(\lambda_n;x,t)}}{\lambda - \lambda_n} \chi\left(\lambda_n^{-1}; x,t\right).
\end{equation}
In particular, the solution is a rational function of the form
\begin{equation} \chi(\lambda;x,t) = 1 + \frac{1}{\pi} \sum_{n = 1}^n \frac{f_n(x,t)}{\lambda - \lambda_n}. \end{equation}
Plugging the rational form of the solution into the algebraic equation (\ref{eq:Nsol}) gives the linear equation 
\begin{equation} \label{eq:Nsollinsys} \sum_{m = 1}^N \left(\delta_{nm}  + \frac{1}{\pi}  \frac{r_n e^{-\tilde \phi(\lambda_n;x,t)}}{\lambda_m - \lambda_n^{-1}}\right) f_m = r_n e^{-\tilde \phi(\lambda_n;x,t)}. \end{equation}
This linear equation can be solved numerically for the values $f_m$, $m = 1,2,\dotsm,N$.
From these values $f_m$ an N-solution $\varphi$, $\eta$ to the Kaup--Broer system can be computed by theorem 1 via
\begin{equation} \label{eq:soldisc} \varphi(x,t) = -\log \left(1 - \frac{1}{\pi} \sum_{n = 1}^\infty \frac{f_n(x,t)}{\lambda_n} \right), \quad \eta(x,t) = -\frac{1}{2} \varphi_{xx} -\frac{1}{\pi} \sum_{n = 1}^N (f_n)_x(x,t).  \end{equation}
We describe these solutions as N-solitons, because as $t \to \pm \infty$ the field $\eta$ that could represent a free surface water wave has the form of N well separated solitary waves.
The corresponding field $\varphi$ limits to a solution with N kinks as $t \to \pm \infty$.
The constituent solitons of $\eta$ have well defined velocities $v_n^{sol} = -( \lambda_n + \lambda_n^{-1} )/2$ when the solitons are well separated and noninteracting.

Consider the contour
\begin{equation} \Gamma = [-\gamma_2,-\gamma_1] \cup [-\gamma_1^{-1},-\gamma_2^{-1}] \cup [\gamma_2^{-1},\gamma_1^{-1}] \cup [\gamma_1,\gamma_2]  \end{equation}
with $1 < \gamma_1 < \gamma_2 < \infty$ oriented from left to right on all components,
and let $R_1$ be a real valued functions on $\Gamma$ that is nonnegative on $[-\gamma_1^{-1},-\gamma_2^{-1}] \cup [\gamma_1, \gamma_2]$ and nonpositive on $[-\gamma_2,-\gamma_1] \cup [\gamma_2^{-1},\gamma_1^{-1}]$.
Moreover, we assume that $R_1$ only vanishes on a finite number of intervals, and that $R_1$ is H\"{o}lder continuous on the support of $R_1$.
The pair $\{\Gamma,R_1\}$ are the spectral data for a class of solutions determined by the dressing functions
\begin{equation} R_0(\lambda) = \int_{\Gamma} \delta(\lambda-s) R_1(s) ds. \end{equation}
The spectral data give rise to the nonlocal $\bar \partial$ problem
\begin{equation} \label{eq:dbar1p1} \frac{\partial \chi}{\partial \bar{\lambda}}(\lambda;x,t) = \int_{\Gamma} \delta(\lambda - s)  e^{-2\tilde \phi(s;x,t)}R_1(s) \chi(s^{-1}) ds,   \end{equation}
and the integral equation giving the solution $\chi$ to the nonlocal $\bar \partial$ problem normalized so that $\chi \to 1$ as $\lambda \to \infty$ is
\begin{equation} \label{eq:dbarint} \chi(\lambda) - \frac{1}{\pi} \int_{\Gamma}  \frac{e^{-2\tilde \phi(s;x,t)}R_1(s) }{\lambda-s} \chi(s^{-1})  ds = 1. \end{equation}
A subtlety here is that when $R_1(s) \ne 0$ and $R_1(s^{-1}) \ne 0$ simultaneously, then $\chi$ can become singular on $\Gamma$.
This means that we should interpret the nonlocal dbar problem (\ref{eq:dbar1p1}) as a jump problem
\begin{equation} \chi_+(s)-\chi_-(s)=R_1(s)e^{-2\tilde\phi(s;x,t)}(\chi_+(s^{-1})+\chi_-(s^{-1})) \end{equation}
where the solution has the form
\begin{equation} \chi(\lambda;x,t) = 1 +  \frac{1}{\pi} \int_\Gamma \frac{f(s)}{\lambda - s} ds. \end{equation}
The function $f$ solves the 1D integral equation 
\begin{equation} \label{eq:inteqf} f(s;x,t) - e^{-2\tilde \phi(s;x,t)} R_1(s) \mathcal{H}_\Gamma f(s^{-1};x,t) = e^{-2\tilde \phi(s;x,t)} R_1(s), \end{equation}
where $\mathcal{H}_\Gamma$ is the Hilbert transform on $\Gamma$ defined by
\begin{equation} \mathcal{H}_\Gamma g(s) = \frac{1}{\pi} \fint_\Gamma \frac{g(s';x,t)}{s - s'} ds'. \end{equation}
Note that the support of the function $f$ solving (\ref{eq:inteqf}) is the same as the support of $R_1$.
Once the solution $f(s)$ to integral equation (\ref{eq:inteqf}) has been computed, a solution $\varphi$, $\eta$ to the Kaup--Broer system can be computed by theorem 1 as
\begin{equation} \label{eq:solcont} \begin{cases} \displaystyle{ \varphi(x,t) = -\log \left(1 - \frac{1}{\pi} \int_{\Gamma} \frac{f(s;x,t)}{s}  ds \right)} \\[1em] \displaystyle{ \eta(x,t) = - \frac{1}{2} \varphi_{xx}(x,t) - \frac{1}{\pi} \int_\Gamma f_x(s;x,t) ds} \end{cases}. \end{equation}
We call these solutions primitive solutions, these are analogous to the primitive solutions to the KdV equation \cite{DZZ16,ZDZ16,ZZD16,NZZ19}.

\begin{thm}
Let $R_1$ be a real valued functions on $\Gamma$ that is nonnegative on $[-\gamma_1^{-1},-\gamma_2^{-1}] \cup [\gamma_1, \gamma_2]$ and nonpositive on $[-\gamma_2,-\gamma_1] \cup [\gamma_2^{-1},\gamma_1^{-1}]$.
Suppose that $f(s,x,t)$ solves (\ref{eq:inteqf}), then 
$\eta$ and $\varphi$ defined from $f$ according to (\ref{eq:solcont}) solve the 1+1D Kaup--Broer system with $a=-1$ and $\sigma = 1$.
\end{thm}

Consider the following property (I) on $R_1$: If $s$ is in the support of $R_1$, then $s^{-1}$ is not in the support of $R_1$.
If $R_1$ satisfies this property, then the principle value integral appearing in equation (\ref{eq:inteqf}) reduces to a regular integral and the equation is a regular Fredholm integral equation of the second kind on the support of $R_1$.
Otherwise, equation (\ref{eq:inteqf}) is a nonlocal singular integral equation.

Suppose that property (I) is satisfied, the integral equation (\ref{eq:inteqf}) is approximated via an $N$ point quadrature rule on the support of $R_1$, and the integrals appearing in (\ref{eq:solcont}) are approximated by the same quadrature rule.
Then the approximations of (\ref{eq:inteqf}) and (\ref{eq:solcont}) by quadrature have the form of (\ref{eq:Nsollinsys}) and (\ref{eq:soldisc}), and thus the approximation to the solution determined by $\{\Gamma,R_1\}$ is an exact $N$-soliton solution.
In this manner, the solution determined by $\{\Gamma,R_1\}$ can be interpreted as an element of the closure of the N-soliton solutions to the Kaup--Broer system in the topology of uniform convergence in compact sets.
If property (I) is not satisfied, then care needs to be taken in dealing with the singularity in defining such a quadrature rule.
However, even when property (I) is not satisfied, the interpretation of the solution as an element of the closure of the N-solitons solutions is still valid.

The primitive solutions to the Kaup--Broer system determined by $\{\Gamma,R_1\}$ can also be computed via the solution to a Riemann--Hilbert problem.
To see this, consider the function $\boldsymbol{\chi}(\lambda) = [\chi(\lambda),\chi(\lambda^{-1})]$ and the contour
\begin{equation} \tilde \Gamma = \left\{ s \in \Gamma : R_1(s) \ne 0 \text{ or } R_1(s^{-1}) \ne 0   \right\}. \end{equation}
The contour $\tilde \Gamma$ consists of $\Gamma$ with a finite number of intervals removed by the assumption that $R_1$ vanishes only on a finite number of intervals in $\Gamma$.
Then $\boldsymbol{\chi}$ solves the following Riemann--Hilbert problem:
\begin{rhp} \label{rhpKB} For all $x,t$ find a $1 \times 2$ vector valued function $\boldsymbol{\chi}(\lambda;x,t)$ such that
\begin{enumerate}
\item $\boldsymbol{\chi}$ is a holomorphic function of $\lambda \in \mathbb{C} \setminus \tilde \Gamma$.
\item The boundary values
\begin{equation}\boldsymbol{\chi}_+(s;x,t) = \lim_{\epsilon \to 0^+} \boldsymbol{\chi}(s+i\epsilon;x,t), \quad \boldsymbol{\chi}_-(s;x,t) = \lim_{\epsilon \to 0^+} \boldsymbol{\chi}(s-i\epsilon;x,t)\end{equation}
of $\boldsymbol{\chi}$ for $s \in \tilde \Gamma \setminus\{\text{endpoints of }\tilde\Gamma\}$ from above or below are continuous.
\item The function $\boldsymbol{\chi}$ has singularities that are less severe than poles on the endpoints of $\tilde \Gamma$.
\item The boundary values $\boldsymbol{\chi}_\pm(\lambda;x,t)$ of $\boldsymbol{\chi}(\lambda;x,t)$ from above and below for $\lambda \in \Gamma$ are related by
\begin{equation}
\boldsymbol{\chi}_+(s;x,t) = \boldsymbol{\chi}_-(s;x,t) V(s;x,y)
\end{equation}
where
\begin{equation} \label{eq:jump} V(s;x,t) =  \begin{pmatrix}  \frac{1 + R_1(s) R_1(s^{-1})}{1 - R_1(s) R_1(s^{-1})} & \frac{2 i  R_1(s^{-1})}{1 - R_1(s) R_1(s^{-1})} e^{2 \tilde \phi(s;x,t)} \\ -\frac{2 i  R_1(s)}{1 - R_1(s) R_1(s^{-1})} e^{-2 \tilde \phi(s;x,t)} & \frac{1 + R_1(s) R_1(s^{-1})}{1 - R_1(s) R_1(s^{-1})}  \end{pmatrix}. \end{equation}
\item The function $\boldsymbol{\chi}$ has the limiting behaviors $\chi_1(\lambda) \to 1$ as $\lambda \to \infty$ and $\chi_2(\lambda) \to 1$ as $\lambda \to 0$.
\item $\boldsymbol{\chi}$ satisfies the symmetry
\begin{equation} \boldsymbol{\chi}(s^{-1};x,t) = \boldsymbol{\chi}(s;x,t) \begin{pmatrix} 0 & 1 \\ 1 & 0 \end{pmatrix}. \end{equation}

\end{enumerate}
\end{rhp}

\begin{prop}
A solution to the Riemann--Hilbert problem has Taylor expansions of the form
\begin{equation} \boldsymbol{\chi}(\lambda;x,t) = (\chi_0^\infty(x,t),\chi_0^0(x,t)) + (\chi_1^\infty(x,t),\chi_1^0(x,t)) \lambda^{-1} + O(\lambda^{-2}), \quad \lambda \to \infty,  \end{equation}
\begin{equation} \boldsymbol{\chi}(\lambda;x,t) = (\chi_0^0(x,t),\chi_0^\infty(x,t)) + (\chi_1^0(x,t),\chi_1^\infty(x,t)) \lambda + O(\lambda^{2}), \quad \lambda \to 0, \end{equation}
where $\chi_0^\infty = 1$.
We can construct primitive solutions to the Kaup--Broer system using the Taylor coefficients via
\begin{equation} \label{eq:KBRHP} \varphi(x,t) = -\log \left(\chi_0^0(x,t) \right), \quad \eta(x,t) = - \frac{1}{2} \varphi_{xx}(x,t) - (\chi_0^0)_x(x,t) \end{equation}
by corollary 2.
\end{prop}

\section{Finite Gap Primitive Solutions}

In this section we construct a family of algebro-geometric finite gap solutions to the Kaup--Broer system that can be constructed as primitive solutions. We thus end up with an effective way to compute sequences of N-soliton solutions converging to these algebro-geometric finite gap solutions.
These are the algebro-geometric finite gap solutions that can be thought of as nonlinear superpositions of counter propagating solutions to the KdV equation.
We work these solutions out in detail for the case $a=-1$ and $\sigma=1$.
This construction is analogous to the construction for the KdV equation appearing in \cite{N19}.

We can consider the hyperelliptic projective curve $\Sigma'$ defined by 
\begin{equation} w^2=P_{4g}(\lambda),  \quad P_{4g}(\lambda)= \prod_{n = 1}^{2g_\ell} (\lambda-\eta_{n}^{-1})(\lambda-\eta_{n}) \prod_{m = 1}^{2g_r}(\lambda+\xi_{m})(\lambda+\xi_{m}^{-1}).  \end{equation}
This curve has genus $g'=2(g_\ell+g_{r})-1$.
Consider the involution
\begin{equation} \iota(\lambda,w)=(\lambda^{-1},-w) \end{equation}
and define the curve $\Sigma = \Sigma'/\left<\iota \right>$ where $\left<\iota \right>$ is the group $\{id.,\iota\}$ which acts on $\Sigma$.
We put the branch cuts of $\sqrt{P_{4g}(\lambda)}$ on $[\eta_{2n-1},\eta_{2n}]$, $[\eta_{2n}^{-1},\eta_{2n-1}^{-1}]$ for $n = 1,\dots,g_\ell$ and $[-\xi_{2n},-\xi_{2n-1}]$, $[-\xi_{2n-1}^{-1},-\xi_{2n}^{-1}]$ for $n = 1,\dots,g_{r}$.
One sheet of $\Sigma'$ gives a coordinate $\lambda \in \mathbb{C} \setminus \tilde \Gamma$ where
\begin{equation} \tilde \Gamma = \bigcup_{n = 1}^{g_\ell} [\eta_{2n-1},\eta_{2n}] \cup [\eta_{2n}^{-1},\eta_{2n-1}^{-1}] \cup \bigcup_{m = 1}^{g_r} [-\xi_{2m},-\xi_{2m-1}] \cup [-\xi_{2m-1}^{-1},-\xi_{2m}^{-1}] \end{equation}
that covers all of $\Sigma$ except $g$ circles.
We will not need to make reference to the points on the interiors of the branch cuts, however the branch points will be important.
If $\lambda_0$ is a branch point of $\sqrt{P_{4g}(\lambda)}$, then $\lambda_0^{-1}$ is also a branch point of $\sqrt{P_{4g}(\lambda)}$.
Moreover, $\lambda_0$ and $\lambda_0^{-1}$ correspond to the same point on $\Sigma$, and we use the notation $\left<\lambda_0\right>$ to refer to this point.
Since $\Sigma'$ is a double covering of $\Sigma$ it follows from the Euler characteristic that $\Sigma$ has genus $g = g_\ell+g_r$.

We will now compute a basis of abelian differentials of the first kind on $\Sigma$ as follows:
Consider the basis
\begin{equation}  \alpha_n' = \frac{\lambda^{n-1} d \lambda}{\sqrt{P_{4g}(\lambda)}}  \end{equation}
for $j={1,2,\dots,2g-1}$ of Abelian differential of the second kind on $\Sigma'$.
The involution $\iota$ acts on the above basis of  abelian differential of the first kind by $\iota^*$ as
\begin{equation} \iota^* \alpha_n'= \frac{\lambda^{2g-n-1} d\lambda}{\sqrt{\lambda^{4g}P_{4g}(\lambda^{-1})}} = \frac{\lambda^{2g-n-1}d\lambda}{\sqrt{P_{4g}(\lambda)}} = \alpha_{2g-n}'\end{equation}
where it is easy to verify that $\lambda^{4g}P_{4g}(\lambda^{-1})$ because each root of $P_{4g}(\lambda)$ is the multiplicative inverse of another root of $P_{4g}(\lambda)$.

The Abelian differentials of the first kind $\alpha_n$ on $\Sigma$ that can be represented as 
\begin{equation} \alpha_n=\alpha_n'+\alpha_{2g-n}'=\frac{\lambda^{n-1}+\lambda^{2g-n-1}}{\sqrt{P_{4g}(\lambda)}} d\lambda \end{equation}
for $n=1,2,\dots,g$ form a basis.
Let $a_j$ and $b_j$ for $j=1,2,\dots,g$ be a canonical homology basis for the first homology group $H_1(\Sigma)$ satisfying $a_i \circ b_j = \delta_{ij}$, $a_i \circ a_j =0$ and $b_i\circ b_j=0$.
We can then form a basis of normalized abelian differentials of the first kind $\omega_n$ satisfying
\begin{equation} \int_{a_j} \omega_i = 2 \pi i \delta_{ij} \end{equation}
as linear combinations of $\alpha_n$.

The normalized basis of abelian differentials of the first kind allows us to define the Able map $\mathbf{A}$ with entries
\begin{equation} A_n(\lambda)= \int_\infty^\lambda \omega_n  \end{equation}
mapping $\Sigma$ into the Jacobi variety.
The Riemann matrix for $\Sigma$ is
\begin{equation} B_{ij} = \int_{b_j} \omega_i .\end{equation}
The vector of Riemann constants $\mathbf{K}$ has entires
\begin{equation} K_j =  \frac{2 \pi i + B_{jj}}{2} - \frac{1}{2 \pi i} \sum_{\ell \ne j} \int_{a_{\ell}} A_j(\lambda) \omega_{\ell}. \end{equation}

We can also define the abelian differentials of the second kind $\omega^{(n)}$ with asymptotic behavior $\omega^{(n)}=(n \lambda^{n-1}+O(1)) d \lambda$ as $\lambda \to \infty$, asymptotic behavior $\omega^{(n)}=(n \lambda^{-n-1}+O(1)) d\lambda$ as $\lambda \to 0$, and
\begin{equation} \int_{a_j} \omega^{(n)}=0. \end{equation}
These differentials are given by
\begin{equation} \omega^{(n)} = \frac{n(\lambda^{2g+n-1} + \lambda^{-n-1})}{\sqrt{P_{4g}(\lambda)}} d \lambda+\sum_{j = 1}^g c_j^{(n)} \omega_j \end{equation}
where
\begin{equation} c_j^{(n)}=\frac{i}{2 \pi} \int_{a_j}\frac{n(\lambda^{2g+n-1} + \lambda^{-n-1})}{\sqrt{P_{4g}(\lambda)}} d \lambda.  \end{equation}
The integrals of these differentials satisfy
\begin{equation} \int_\infty^\lambda \omega^{(n)}= \lambda^n+O(1), \text{ as }\lambda \to \infty; \quad \int_\infty^\lambda \omega^{(n)}= -\lambda^{-n}+O(1), \text{ as } \lambda \to 0. \end{equation}
An important aspect of $\omega^{(n)}$ is the vector $\boldsymbol{\Omega}^{(n)}$ with entries
\begin{equation} \Omega_j^{(n)} =\int_{b_j} \omega^{(n)}, \text{ for } j = 1,2,\dots,g. \end{equation}

Let us pick a degree $g$ divisor $\mathfrak{P}$ as a direct sum of points
\begin{align} &P_1 \in [\eta_1^{-1},\eta_1], \\
&P_j \in [\eta_{2j+1}^{-1},\eta_{2j}^{-1}]\cup[\eta_{2j},\eta_{2j+1}] \text{ for } j = 2,3,\dots,g_\ell, \\
& P_{g_\ell+1} \in [-\xi_1,-\xi_1^{-1}], \\
& P_{g_\ell+j} \in [-\xi_{2j+1},-\xi_{2j}] \cup [-\xi_{2j}^{-1},-\xi_{2j+1}^{-1}] \text{ for } j = 2,3,\dots,g_r \end{align}
This divisor is nonspecial because all points are distinct \cite{BBEIM94,TD13}.

The Baker--Akheizer 2-point function for the Kaup Broer system is the unique meromorphic function $\psi$ on $\Sigma$ such that 
$\psi$ has simple poles on the points of $\mathfrak{P}$ and asymptotic behavior
\begin{equation} \psi(\lambda;x,t) = e^{\tilde \phi(\lambda;x,t)}(1+O(\lambda^{-1})), \lambda \to \infty \text{ and } \psi(\lambda;x,t) = e^{\tilde \phi(\lambda;x,t)}O(1), \lambda \to 0.  \end{equation}
The function $e^{-\tilde \phi(\lambda;x,t)} \psi(\lambda;x,t)$ is holomorphic at $\infty$ and $0$ respectively, and therefore $\psi(\lambda;x,t)$ has the following uniformly convergent expansions:
\begin{equation} \label{eq:expinf} \psi(\lambda;x,t) = e^{\tilde \phi(\lambda;x,t)}\left(1+\sum_{n=1}^\infty \chi_n^\infty(x,t) \lambda^{-n} \right) \end{equation}
as $\lambda \to \infty$ and \begin{equation} \label{eq:exp0} \psi(\lambda;x,t) =  e^{\tilde \phi(\lambda;x,t)} \sum_{n = 0}^\infty \chi_n^0(x,t) \lambda^n\end{equation} as $\lambda \to 0$.

\begin{rmk}
A standard application of the Riemann--Roch theorem tells us that the space of functions satisfying all properties of the Baker--Akheizer function except we allow the more general asymptotic behavior $\psi(\lambda;x,t) = e^{\tilde \phi(\lambda;x,t)}O(1)$, $\lambda \to \infty$ is one dimensional \cite{D81}. \label{rmk:unique}
\end{rmk}

Therefore the Baker--Akheizer 2-point function is unique.
In fact, the Baker--Akheizer function has the explicit formula
\begin{equation} \psi(\lambda;x,t) = \exp\left(\frac{1}{2}\int_\infty^\lambda \omega^{(1)} x + \frac{1}{4} \int_\infty^\lambda \omega^{(2)} t \right) \frac{\theta(\mathbf{A}(\lambda)+\boldsymbol{\Omega}^{(1)} x + \boldsymbol{\Omega}^{(2)} t - \mathbf{A}(\mathfrak{P}) -\mathbf{K};B)\theta( - \mathbf{A}(\mathfrak{P}) -\mathbf{K};B)}{\theta(\mathbf{A}(\lambda) - \mathbf{A}(\mathfrak{P}) -\mathbf{K};B)\theta(\boldsymbol{\Omega}^{(1)} x + \boldsymbol{\Omega}^{(2)} t - \mathbf{A}(\mathfrak{P}) -\mathbf{K};B)}, \end{equation}
where $\theta(\mathbf{z};B)$ is the Riemann theta function
\begin{equation} \theta(\mathbf{z},B) = \sum_{\mathbf{n} \in \mathbb{Z}^g} \exp \left( {\frac{1}{2} \mathbf{n} \cdot B \mathbf{n} + \mathbf{n} \cdot \mathbf{z}} \right). \end{equation}

\begin{rmk} This representation of the Kaup--Broer system spectral curve is different than that of Matveev--Yavor \cite{MY79}.
The spectral curve $\Sigma$ is isomorphic to a spectral curve of the form considered by Matveev--Yavor.
Let us consider the projective curve $\Sigma_{MY}$ defined by
\begin{equation} r^2 = P_{2g+2}(k), \quad P_{2g+2}(k)= (k^2-1) \prod_{n=1}^{g_\ell} (k - \eta_n) \prod_{n=1}^{g_r} (k + \xi_n)=P_{2g+2}(k). \end{equation}
The mapping
\begin{equation} h(\lambda) = (k(\lambda),r(\lambda)), \quad k(\lambda)=\frac{\lambda+\lambda^{-1}}{2}, \quad r = \text{\normalfont sgn} \circ \log(|\lambda|) \sqrt{P_{2g+2}\left(k
(\lambda)\right)} \end{equation}
maps $\Sigma$ bijectively onto $\Sigma_{MY}$.
Therefore, $\Sigma_{MY}$ and $\Sigma$ are isomorphic.
If $\pi$ is the projection $\pi(k,r)=k$ then $\pi(h(\mu_n)) \in [\eta_{2n-1},\eta_{2n}]$ and $\pi(h(\mu_{g_\ell+n})) \in [-\xi_{2n},-\xi_{2n-1}]$.
Therefore, our spectral data satisfied the reality and regularity conditions determined by Smirnov \cite{S86}. \end{rmk}

Recall that $\tilde L$ and $\tilde M$ were defined in remark \ref{rmk:linsysred}.
We set the potentials in $\tilde L$ and $\tilde M$ according to the (\ref{eq:potentials}) except we use coefficients from (\ref{eq:expinf},\ref{eq:exp0}) and the $u$ and $v$ derivatives reduce to $x$ derivatives.
The functions $e^{-\tilde \phi(\lambda;x,t)}\tilde L\psi(\lambda)$ and $e^{-\tilde \phi(\lambda;x,t)}\tilde M\psi(\lambda)$ are regular at $0$ and satisfy the asymptotic conditions  $e^{-\tilde \phi(\lambda;x,t)} \tilde L\psi(\lambda) \to 0$ and $e^{-\tilde \phi(\lambda;x,t)} \tilde M\psi(\lambda) \to 0$, $\lambda \to \infty$.
Moreover, it can be verified that $\tilde L\psi(\lambda)$ and $\tilde M \psi(\lambda)$ satisfy all the properties of the Baker--Akheizer function except we allow the more general asymptotic behavior $\psi(\lambda;x,t) = e^{\tilde \phi(\lambda;x,t)}O(1)$, $\lambda \to \infty$.
This means $\tilde L\psi =0$ and $\tilde M\psi=0$ by remark \ref{rmk:unique}.  
Remark \ref{rmk:linsysred} then leads to the following proposition:
\begin{prop}
We can construct solutions to the Kaup--Broer system using asymptotic expansion (\ref{eq:expinf},\ref{eq:exp0}) by taking
\begin{equation} \varphi(x,t) = -\log \left(\chi_0^0(x,t) \right), \quad \eta(x,t) = - \frac{1}{2} \varphi_{xx}(x,t) - (\chi_0^0)_x(x,t). \end{equation}
\end{prop}
This argument is analogous to the corresponding argument justifying the $\bar \partial$ dressing method.

\begin{rmk} The $g$ values $\eta_n, \xi_n$ determining $\Sigma$ and the degree $g$ divisor $\mathfrak{P}$ on $\Sigma$ constitute the spectral data for the finite gap solutions.
The positive spectral data corresponds to a left moving finite gap KdV scattering problem, and the negative  spectral data corresponds to a right moving finite gap KdV scattering problem. \end{rmk}

Let us introduce an auxiliary Baker--Akheizer function $\psi^{(aux)}(\lambda)$ set by the zero divisor $\mathfrak{P}$ and pole divisor $\mathfrak{P}_0$ consisting of a direct sum of the points $\left<\eta_{2j-1} \right>$ and $\left<-\xi_{2j-1}\right>$.
This function is given by
\begin{equation} \psi(\lambda;x,t) = \exp\left(\sum_{n = 1}^g\int_\infty^\lambda \omega^{(n)} t_n \right) \frac{\theta(\mathbf{A}(\lambda) - \mathbf{A}(\mathfrak{P}_0) -\mathbf{K};B)\theta( - \mathbf{A}(\mathfrak{P}) -\mathbf{K};B)}{\theta(\mathbf{A}(\lambda) - \mathbf{A}(\mathfrak{P}) -\mathbf{K};B)\theta(- \mathbf{A}(\mathfrak{P}_0) -\mathbf{K};B)} \end{equation}
where the coefficient $t_j$ are determined by solving 
\begin{equation} \Omega \mathbf{t} \equiv \mathbf{A}(\mathfrak{P})-\mathbf{A}(\mathfrak{P}_0) \end{equation}
where $\Omega$ is the matrix with entries $\Omega_{jn} = \Omega_j^{(n)}$, and $\equiv$ is equivalence within the Jacobi variety. 


This construction is analogous to the analogous construction used by Trogdon and Deconinck in studying the KdV equation.
The construction of Trodon and Deconinck was used in \cite{TD13} to compute finite gap primitive solutions to the KdV equation.

\begin{prop}
The matrix $\Omega$ is invertible and the solution $\mathbf{t}$ has real entries $t_n$.
\end{prop}

\begin{proof}

The proof of this theorem is similar to the analogous proof for the KdV equation given in \cite{TD13}.
The proof of invertibility is based off arguments presented in \cite{D09}.

Let $\epsilon < \min\{\eta_1,\xi_1,\eta_{2g_\ell}^{-1},\xi_{2g_r}^{-1}\}$.
Let $C_0$ be the circle of radius $\epsilon$ centered at $0$ oriented counterclockwise, and let $C_\infty$ be the circle of radius $\epsilon^{-1}$ cantered at $0$ oriented clockwise.
Let $D_0$ be the disc with boundary $C_0$ containing $0$, and let $D_\infty$ be the disc with boundary $C_\infty$ containing $\infty$.
Inversion maps $C_0$ into $C_\infty$ and vise versa.
Let us choose the $a_j$, $b_j$ so that they avoid $D_0$ and $D_\infty$, and let $\lambda_0$ be a point on $\Sigma$ that is in neither $D_0$ nor $D_\infty$ and avoid $a_j$, $b_j$.
Using stokes theorem and the fact primitives of abelian differentials are single valued on $\Sigma \setminus \bigcup_j a_j \cup b_j$, we can compute \cite{D09}
\begin{equation}0= \iint_{\Sigma} \alpha_m \wedge \omega^{(n)} = \sum_{j = 1}^g \int_{a_j} \alpha_m \int_{b_j} \omega^{(n)} - \int_{C_0}\left( \int_{\lambda_0}^\lambda \alpha_m \right) \omega^{(n)} -\int_{C_\infty}\left( \int_{\lambda_0}^\lambda \alpha_m \right) \omega^{(n)}, \end{equation}
or equivalently
\begin{equation}\Omega_{jn}= \sum_{j = 1}^g A_{jm}^{-1} D_{mn} \end{equation}
where
\begin{equation} A_{mj} = \int_{a_j} \alpha_m, \quad  D_{mn} = \int_{C_0}\left( \int_{\lambda_0}^\lambda \alpha_m \right) \omega^{(n)} + \int_{C_\infty}\left( \int_{\lambda_0}^\lambda \alpha_m \right) \omega^{(n)}. \end{equation}
It is well known that a matrix of a-periods of any basis of abelian differentials of the first kind on a Riemann surface is invertible, therefore if $D_{mn}$ is invertible then $\Omega_{jn}$ is invertible.

The first integral appearing in $D_{mn}$ can be computed using the calculus of residues as
\begin{align} \int_{\tilde C_0}\left( \int_{\lambda_0}^\lambda \alpha_m \right) \omega^{(n)} & = \int_{C_0} \int_{\lambda_0}^\lambda \frac{\lambda_1^{m-1}+\lambda_1^{2g-m-1}}{\sqrt{P_{4g}(\lambda_1)}} d\lambda_1 \frac{n\lambda^{-n-1}}{\sqrt{P_{4g}(\lambda)}} d \lambda \\
& = \frac{2\pi i}{(n-1)!} \left. \frac{d^{n-1}}{d\lambda^{n-1}}  \frac{\lambda^{m-1}+\lambda^{2g-m-1}}{\sqrt{P_{4g}(\lambda)}}\right|_{\lambda = 0}   \end{align}
where we only need to keep the singular part of $\omega^{(n)}$ because the integral involving the regular part is 0.
We have made use of the fact that $P_{4g}(0)=1$.
The second integral appearing in $D_{mn}$ is seen to be equal to the first because it is represented in the $z$-plane as
\begin{align} \int_{ C_\infty}\left( \int_{\lambda_0}^\lambda \alpha_m \right) \omega^{(n)} & = -\int_{C_0} \int_{P_0}^{z^{-1}} \frac{\lambda_1^{n-1}+\lambda_1^{2g-n-1}}{\sqrt{P_{4g}(\lambda_1)}} d\lambda_1 \frac{n z^{-n-1}}{\sqrt{P_{4g}(z)}} dz \\
&  = \int_{C_0} \int_{P_0}^{z} \frac{z_1^{n-1}+z_1^{2g-n-1}}{\sqrt{P_{4g}(z_1)}} dz_1 \frac{n z^{-n-1}}{\sqrt{P_{4g}(z)}} dz. \end{align}
Therefore,
\begin{equation}D_{mn} = \frac{4 \pi i}{(n-1)!} \left. \frac{d^{n-1}}{d\lambda^{n-1}}  \frac{\lambda^{m-1}+\lambda^{2g-m-1}}{\sqrt{P_{4g}(\lambda)}}\right|_{\lambda = 0}. \end{equation}
We see that $D_{mn} = 0$ for $n<m$, $D_{nn} = 4\pi i/(n-1)!$ for $n =1,2,\dots,g-1$ and $D_{gg}=8\pi i/(g-1)$.
The matrix $D_{mn}$ is invertible because it is triangular with nonzero entries on the diagonal.

We will prove that $t_n$ are real by choosing a canonical homology basis $a_j$ and $b_j$ so that $\Omega$ is real, and choosing contours for the Abel map so that $\mathbf{A}(\mathfrak{P})-\mathbf{A}(\mathfrak{P}_0)$ is real.

Let $b_1$ be the contour traversing $[\eta_1^{-1},\eta_1]$ oriented from right to left.
Let $b_n$ for $n = 2,3,\dots,g_\ell$ be the contour traversing $[\eta_{2n-2},\eta_{2n}]$ and then traversing $[\eta_{2n}^{-1},\eta_{2n-2}^{-1}]$ oriented from right to left (while this is disjoint in the $\lambda$ plane, it is a single connected loop on the Riemann surface $\Sigma$).
Let $b_{g_\ell+1}$ be the contours traversing $[-\xi_1,-\xi_1^{-1}]$ oriented from right to left.
Let $b_{g_\ell+n}$ for $n = 2,3,\dots,g_r$ be the contour traversing $[-\xi_{2n-1},-\xi_{2n-2}]$ and then traversing $[-\xi_{2n-2}^{-1},-\xi_{2n-1}^{-1}]$ oriented from right to left.
Let $a_n$ for $n = 1,2,\dots, g_\ell$ be the contour looping once $[\eta_{2g_\ell}^{-1},\eta_{2n-1}^{-1}]$ counterclockwise that can be deformed down to this interval.
Let $a_{g_\ell+n}$ for $n = 1,2,\dots,g_r$ be the contour looping once around $[-\xi_{2g_r},-\xi_{2n-1}]$ counter clockwise that can be deformed down to this interval.

The integral of $\alpha_j$ around $a_i$ can be written as integrals on the bottom and top of branch cuts of $\sqrt{P_{4g}(\lambda)}$ on which $\sqrt{P_{4g}(\lambda)}$ is purely imaginary. This means that $\int_{a_i} \alpha_j$ are purely imaginary because $\alpha_j$ are purely imaginary on the top and bottom of the branch cuts.
Therefore, the coefficients for expanding the normalized differentials $\omega_j$ in terms of $\alpha_n$ are purely real.
The function $\sqrt{P_{4g}(\lambda)}$ is purely real on $b_i$ so $\omega^{(n)}$ is purely real on $b_i$.
Therefore the entires $\int_{b_i} \omega^{(n)}$ of $\Omega$ are purely real.

We can write
\begin{equation} A_j(\mathfrak{P})-A_j(\mathfrak{P}_0) = \sum_{n = 1}^{g_\ell} \int_{\left<\eta_{2n-1}\right>}^{P_n} \omega_j + \sum_{m=1}^{g_r} \int_{\left<-\xi_{2m-1}\right>}^{P_{g_\ell+m}} \omega_j. \end{equation}
We can choose the contour from $\left<\eta_{2n-1}\right>$ to $P_n$ so that it stays in $[\eta_1^{-1},\eta_1]$ for $n=1$ or $[\eta_{2n-1}^{-1},\eta_{2n-2}^{-1}] \cup[\eta_{2n-2},\eta_{2n-1}]$ for $n=2,3,\dots,g_\ell$.
We can choose the contour from $\left<-\xi_{2m-1}\right>$ to $P_{g_\ell+m}$ so that it stays in $[-\xi_1,-\xi_1^{-1}$ for $m=1$ or $[-\xi_{2m-1},-\xi_{2m-2}] \cup [-\xi_{2m-2}^{-1},-\xi_{2m-1}^{-1}]$ for $m=2,3,\dots,g_{r}$.
The differentials $\alpha_j$ are real on these intervals which means $\omega_j$ are real on these intervals too.
Therefore, $\mathbf{A}(\mathfrak{P})-\mathbf{A}(\mathfrak{P}_0)$ is real for this choice of contours.
\end{proof}

We define the function
\begin{equation} \chi(\lambda;x,t) = \xi(\lambda) e^{- \tilde \phi(\lambda;x,t) - \phi^{(aux)}(\lambda)} \psi^{(aux)}(\lambda;x,t) \psi(\lambda;x,t)  \end{equation}
where
\begin{equation} \xi(\lambda) = \prod_{n = 1}^{g_\ell} \frac{\sqrt[4]{\lambda - \eta_{2n-1}}\sqrt[4]{\lambda - \eta_{2n-1}^{-1}}}{\sqrt[4]{\lambda - \eta_{2n}}\sqrt[4]{\lambda - \eta_{2n}^{-1}}} \prod_{n = 1}^{g_r} \frac{\sqrt[4]{\lambda + \xi_{2n-1}}\sqrt[4]{\lambda + \xi_{2n-1}^{-1}}}{\sqrt[4]{\lambda + \xi_{2n}}\sqrt[4]{\lambda + \xi_{2n}^{-1}}}\end{equation}
where we choose $\xi$ to have branch cuts on $\tilde \Gamma$ and asymptotic behaviors $\xi = 1+O(\lambda^{-1})$ as $\lambda \to \infty$.

\begin{thm}
The vector valued function $\boldsymbol{\chi}(\lambda;x,t)=[\chi(\lambda;x,t),\chi(\lambda^{-1};x,t)]$ solves Riemann--Hilbert problem \ref{rhpKB} with $R_1(s)=e_1(s) f_1(s)$ where
\begin{equation} e_1(s) = \exp\left(-2\sum_{n = 1}^{g} t_n (s^n-s^{-n}) \right)\end{equation}
\begin{equation} f_1(s)=  \sum_{n = 1}^{g_\ell} \left( \mathbbm{1}_{[\eta_{2n-1},\eta_{2n}]}(s) - \mathbbm{1}_{[\eta_{2n}^{-1},\eta_{2n-1}^{-1}]}(s)\right) +  \sum_{m=1}^{g_{r}} \left(   \mathbbm{1}_{[-\xi_{2m-1}^{-1},-\xi_{2m}^{-1}]}(s) - \mathbbm{1}_{[-\xi_{2m},-\xi_{2m-1}]}(s) \right).\end{equation}
These finite gap solutions can alternatively be computed from $R_1(s)$ by solving (\ref{eq:inteqf}) for $f(s,x,t)$ and then computing a solution to the Kaup--Broer system from $f(s,x,t)$ using (\ref{eq:solcont}).
\end{thm}

\begin{proof}
The proof of this theorem is a straight forward verification of the properties of the Riemann--Hilbert problem.
The steps differ only minimally from the proof of the analogous result for the KdV equation appearing in \cite{N19}.
We will verify conditions (3) and (4). The proofs of the other conditions are straight forward, and we leave their verification to the reader.

We will say a singularity has order $-p$ we mean it locally behaves like $z^{-p}$.
The poles of $\tilde \psi(\lambda)$ on $\Sigma$ at $\left<\eta_{2n-1} \right>$ and $\left< -\xi_{2m-1} \right>$ manifest themselves in the $\lambda$ plane as order $-\frac{1}{2}$ singularities at $\eta_{2n-1},\eta_{2n-1}^{-1}$ and $-\xi_{2m-1}, -\xi_{2m-1}^{-1}$.
This means that $\chi$ has order $-\frac{1}{4}$ singularities on the endpoints of $\tilde \Gamma$.
These are less severe than poles verifying condition (3).

Let us write $\tilde\psi(s) = \psi^{(aux)}(s) \psi(s;x,t)$, which is a meromorphic function on $\Sigma \setminus\{0,\infty\}$.
Its representation in the $\lambda$ plane therefore satisfy the nonlocal jump relations
\begin{equation} \tilde\psi_+(s)=\tilde \psi_+(s^{-1}), \quad \tilde \psi_-(s) = \tilde \psi_-(s^{-1})\end{equation}
for $s \in \tilde \Gamma$.
The function $\xi$ satisfies the nonlocal jump relations
\begin{equation} \xi_+(s) = -i\xi_+(s^{-1}), \quad \xi_-(s)=-i\xi_-(s^{-1}) \end{equation}
for $s\in [-\xi_{2m-1}^{-1},-\xi_{2m}^{-1}]$ or $s \in[\eta_{2n-1},\eta_{2n}]$.
Let $\tilde e(s) = e^{-\tilde \phi(s)-\phi^{(aux)}(s)}$, then $\tilde e(s)$ satisfies $\tilde e(s)= e^{-2\tilde \phi(s)-2\phi^{(aux)}(s)} \tilde e(s^{-1})$ for all $s$, and in particular for $s \in \tilde \Gamma$.
Putting these nonlocal jump relations together gives the nonlocal jump conditions
\begin{equation} \chi_+(s;x,t)=-ie^{-2\tilde \phi(s;x,t)} R_1(s) \chi_+(s^{-1};x,t), \quad \chi_-(s^{-1};x,t)=ie^{2\tilde \phi(s;x,t)} R_1(s^{-1}) \chi_-(s;x,t)  \end{equation}
for $s\in \tilde \Gamma$, where we use $R_1(s)=e_1(s)f_1(s)$.
Using $\boldsymbol{\chi}_+(s;x,t) = [\chi_+(s;x,t),\chi_-(s^{-1};x,t)]$ and $\boldsymbol{\chi}_-(s) = [\chi_-(s;x,t),\chi_+(s^{-1};x,t)]$ we see that $\boldsymbol{\chi}_+(s;x,t) = \boldsymbol{\chi}_-(s;x,t) V(s;x,t)$ where
\begin{equation} V(s;x,t) = \begin{pmatrix} 0 & iR_1(s)e^{2\tilde \phi(s;x,t)} \\ -iR_1(s)e^{-2\tilde \phi(s;x,t)} & 0 \end{pmatrix}.  \end{equation}
Using the fact that $R_1(s^{-1})=-R_1(s)^{-1}$ for the choice of dressing function $R_1(s)=e_1(s)f_1(s)$, we see the jump matrix $V(s;x,t)$ is equal to the jump matrix given by (\ref{eq:jump}).
\end{proof}

\section{Numerical Primitive Solutions}

In this section we compute numerical primitive solutions to the Kaup--Broer system determined by the continuum spectral data with $\gamma_1 = \tfrac{9}{8}$, $\gamma_2 = 2$ using various functions $R_0(s)$ satisfying property (I).
These particular solutions appear to describe counter propagating dispersive shockwaves with space time regions that appear to be described by finite gap solutions with various numbers of periods.
We believe that these solutions are locally described by the finite-gap solutions discussed in the previous section.
In producing solutions to the Kaup--Broer system in this section we will make use of the invariance of the Kaup--Broer system with respect to the time reversal symmetry $\eta(x,t) \to \eta(x,-t)$, $\varphi(x,t) \to -\varphi(x,-t)$.

\subsection{Case 1:}
The case of
\begin{equation} R_1 (s) = \pi e^{-10} \left(-\mathbbm{1}_{\left[-\tfrac{7}{4},-\tfrac{5}{4}\right]}(s) + \mathbbm{1}_{\left[\tfrac{3}{2},2\right]}(s)\right)  \end{equation}
where if $I$ is an interval then $\mathbbm{1}_I(s)$ is the indicator function of the interval $I$.
When numerically solving (\ref{eq:inteqf}) and numerically computing the integrals in (\ref{eq:solcont}) we use Gauss-Legendre quadrature with 50 points on each interval $[-\tfrac{7}{4},-\tfrac{5}{4}]$ and $[\tfrac{3}{2},2]$.
The solutions were computed on a uniform space time grid, and the $x$ derivatives appearing in (\ref{eq:solcont}) were computed spectrally with fast Fourier transforms.
The matrix approximation of the integral equation (\ref{eq:inteqf}) has a large condition number in this case, and was solved in Mathematica with 90 digits of precision.
From the solutions $\eta,\varphi$ produced by the numerical approximation to (\ref{eq:solcont}), we define the time reversed solution $\eta_1(x,t) = \eta(x,-t)$, $\varphi_1(x,t) = -\varphi(x,-t)$ to the Kaup---Broer system.
While this is a numerical approximation to (\ref{eq:solcont}), it can also be interpreted as a 100-soliton solution (\ref{eq:soldisc}).

Spacetime plots for the solutions $\eta_1$ and $\varphi_1$ are provided in figures \ref{f11} and \ref{f12}, and spatial plots of these solutions are provided in figures \ref{f13} and \ref{f14}.
At early times we see that the $\eta_1$ component of this solution approximates counter propagating dispersive shockwaves with 1-period trailing waves.
These dispersive shockwaves collide, and when the 1-period trailing waves begin to interact we see what we believe to be described by a 2-period solution to the Kaup--Broer system.

\subsection{Case 2:}

We now consider the case of
\begin{equation} R_1 (s) = \pi e^{-10}\left(-\mathbbm{1}_{\left[-\tfrac{15}{8},-\tfrac{13}{8}\right]}(s) - \mathbbm{1}_{\left[-\tfrac{6}{4},-\tfrac{5}{4}\right]}(s)  + \mathbbm{1}_{\left[\tfrac{11}{8},\tfrac{13}{8}\right]}(s) + \mathbbm{1}_{\left[\tfrac{7}{4},2\right]}(s)\right).  \end{equation}
When numerically solving (\ref{eq:inteqf}) and numerically computing the integrals in (\ref{eq:solcont}) we use Gauss-Legendre quadrature with 25 points on each of the intervals $[-\tfrac{15}{8},-\tfrac{13}{8}]$, $[-\tfrac{6}{8},-\tfrac{5}{4}]$, $[\tfrac{11}{8},\tfrac{13}{8}]$ and $[\tfrac{7}{4},2]$.
The solutions were computed on a uniform space time grid, and the $x$ derivatives appearing in (\ref{eq:solcont}) were computed spectrally with fast Fourier transforms.
The matrix approximation of the integral equation (\ref{eq:inteqf}) has a large condition number in this case, and was solved in Mathematica with 90 digits of precision.
From the solutions $\eta,\varphi$ produced by the numerical approximation to (\ref{eq:solcont}), we define the time reversed solution $\eta_2(x,t) = \eta(x,-t)$, $\varphi_2(x,t) = -\varphi(x,-t)$.
The solution $\eta_2$, $\varphi_2$ can be interpreted as a 100-soliton solution to the Kaup--Broer system.

Spacetime plots for the solutions $\eta_2$ and $\varphi_2$ are provided in figures \ref{f21} and \ref{f22}, and spatial plots of these solutions are provided in figures \ref{f23} and \ref{f24}.
At early times we see that the $\eta_2$ component of this solution approximates counter propagating dispersive shockwaves with trailing waves that we believe to be described by 2-period solutions to the Kaup--Broer system.
When these solutions collide, we initially see what appears to be a region where counter propagating modulated 1-period solutions are interacting.
We then see regions with modulated 1-period solutions interacting with the 2-period trailing waves in a region where we believe the solutions can be described by modulated 3-period solutions to the Kaup--Broer system.
There also appear to be space time regions where the 2-period trailing waves are interacting in what we believe to be described by a 4-period solution to the Kaup--Broer system.

\subsection{Case 3:}

The case of
\begin{equation} R_1 (s) = \pi e^{-10}\left(-\mathbbm{1}_{\left[-2,-\tfrac{15}{8}\right]}(s) - \mathbbm{1}_{\left[-\tfrac{7}{4},-\tfrac{13}{8}\right]}(s) - \mathbbm{1}_{\left[-\tfrac{11}{8},-\tfrac{5}{4}\right]}(s) + \mathbbm{1}_{\left[\tfrac{11}{8},\tfrac{13}{8}\right]}(s) + \mathbbm{1}_{\left[\tfrac{7}{4},2\right]}(s)\right).  \end{equation}
When numerically solving (\ref{eq:inteqf}) and numerically computing the integrals in (\ref{eq:solcont}) we use Gauss-Legendre quadrature with 30 points on each of the intervals $[-2,-\tfrac{15}{8}]$, $[-\tfrac{7}{4},-\tfrac{13}{8}]$,  $[-\tfrac{11}{8},-\tfrac{5}{4}]$, $[\tfrac{11}{8},\tfrac{13}{8}]$ and $[\tfrac{7}{4},2]$.
The solutions were computed on a uniform space time grid, and the $x$ derivatives appearing in (\ref{eq:solcont}) were computed spectrally with fast Fourier transforms.
The matrix approximation of the integral equation (\ref{eq:inteqf}) has a large condition number in this case, and was solved in Mathematica with 100 digits of precision.
From the solutions $\eta,\varphi$ produced by the numerical approximation to (\ref{eq:solcont}), we define the time reversed solution $\eta_3(x,t) = \eta(x,-t)$, $\varphi_3(x,t) = -\varphi(x,-t)$.
The solution $\eta_3$, $\varphi_3$ can be interpreted as a 150-soliton solution to the Kaup--Broer system.

Spacetime plots for the solutions $\eta_3$ and $\varphi_3$ are provided in figures \ref{f31} and \ref{f32}, and spatial plots of these solutions are provided in figures \ref{f33} and \ref{f34}.
This solution has a fairly complicated structure.
We believe this solution has regions described by modulated $g$-period solutions to the Kaup--Broer with $g$ ranging from 0 thru 5 ($g = 0$ corresponding to constant solutions).

\vfill

\pagebreak
\begin{figure}[h]
\center
\includegraphics[width=0.7 \textwidth]{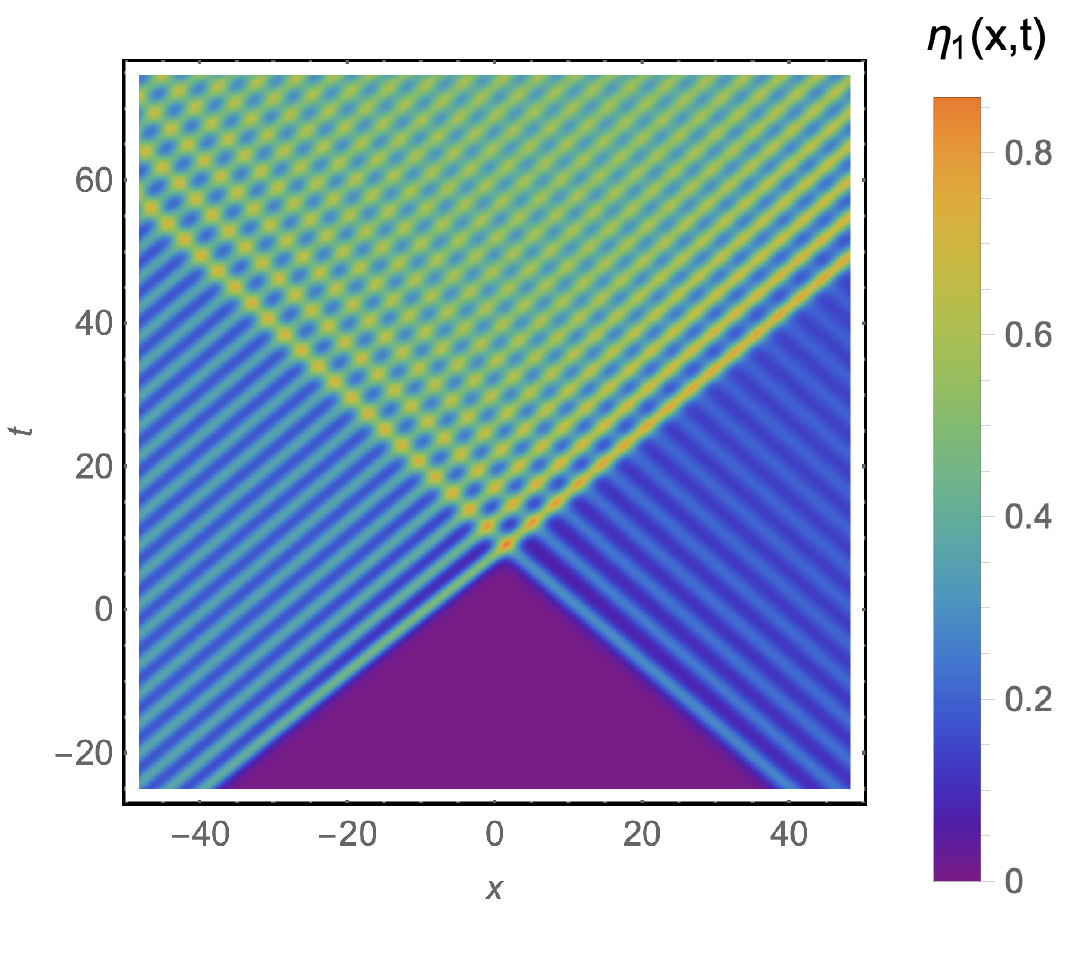}
\caption{A spacetime plot for the field $\eta_1$ for the 100-soliton case 1 solution to the Kaup--Broer system.}
\label{f11}
\end{figure}

\begin{figure}[h]
\center
\includegraphics[width=0.7 \textwidth]{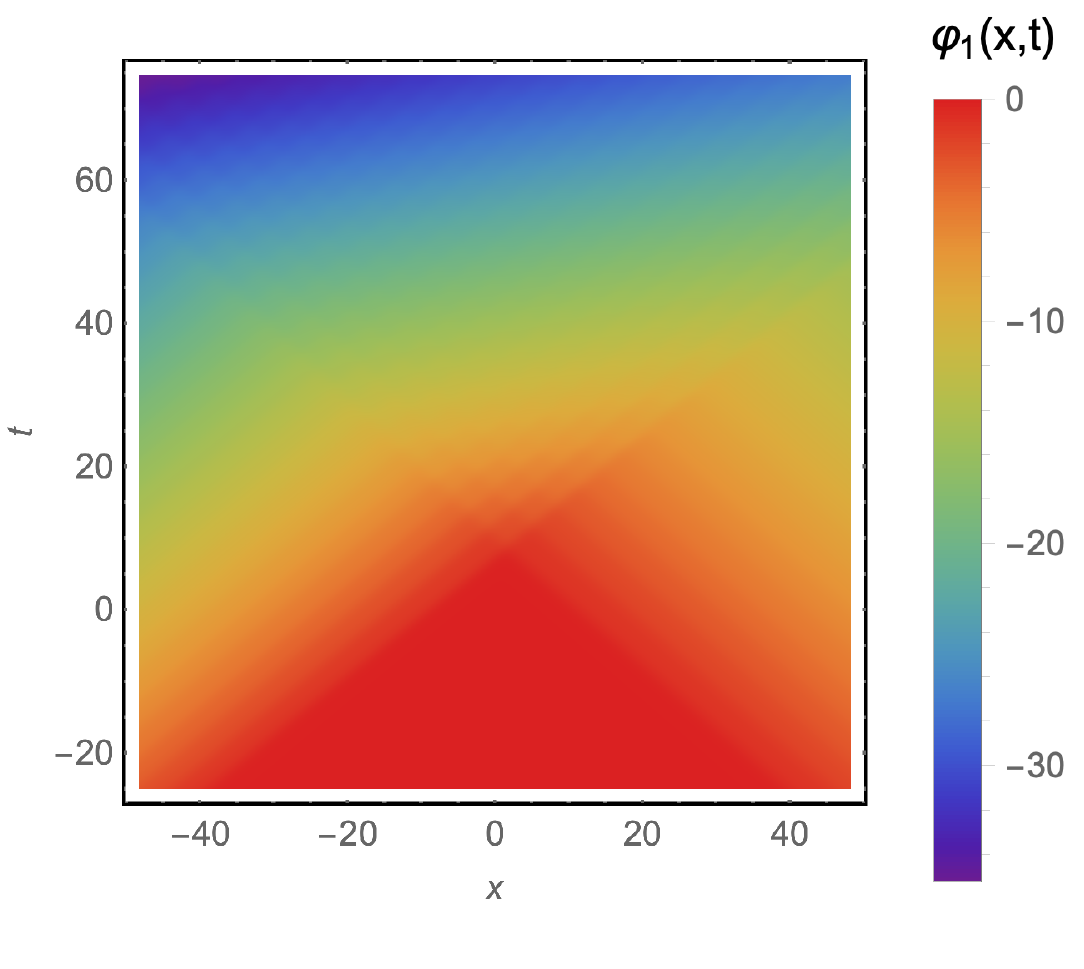}
\caption{A spacetime plot for the field $\varphi_1$ for the case 1 solution to the Kaup--Broer system.}
\label{f12}
\end{figure}

\begin{figure}[h]
\center
\includegraphics[width=0.65 \textwidth]{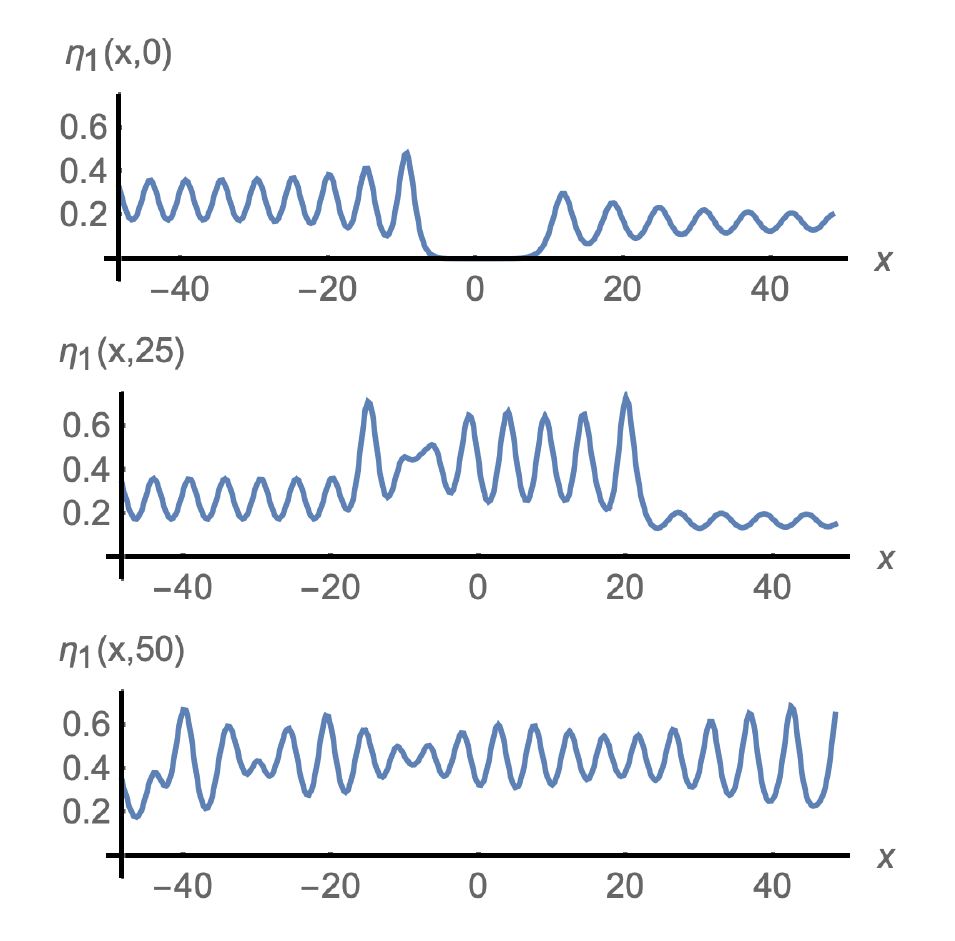}
\caption{Spatial plots for the field $\eta_1$ for the case 1 solution to the Kaup--Broer systems plotted at times $t=0$, $t=25$ and $t = 50$.}
\label{f13}
\end{figure}

\begin{figure}[h]
\center
\includegraphics[width=0.65 \textwidth]{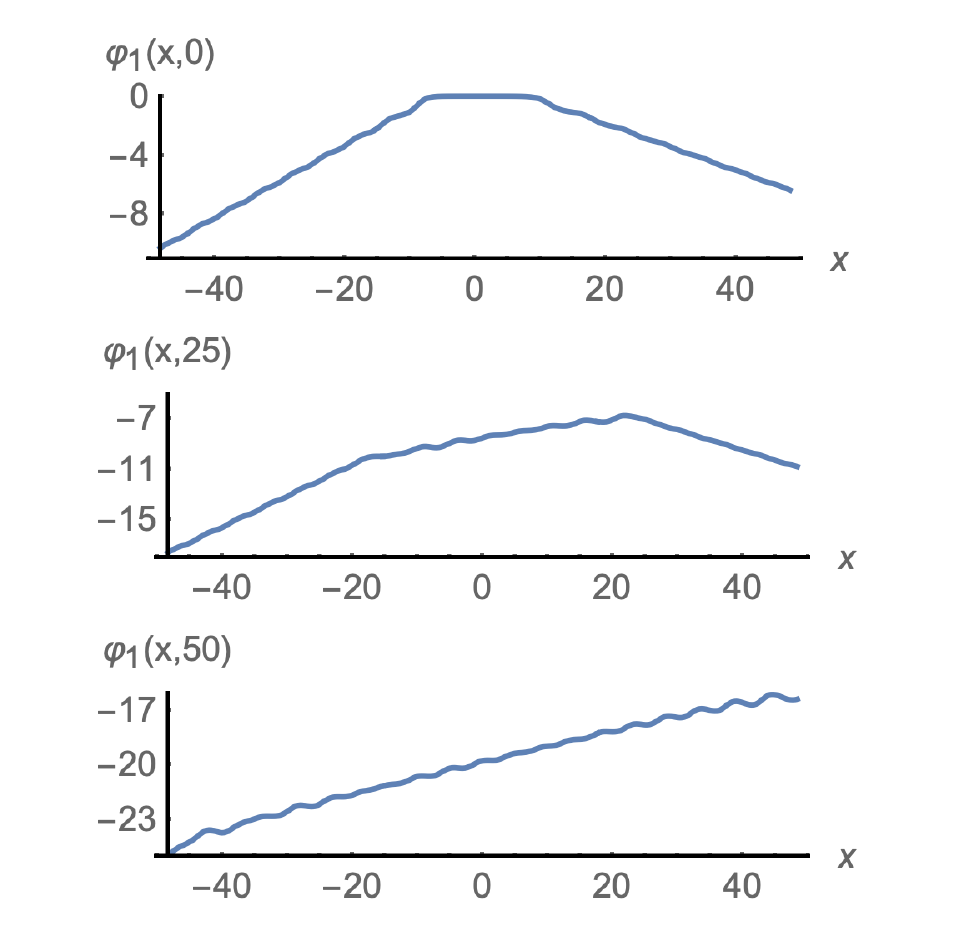}
\caption{Spatial plots for the field $\varphi_1$ for the case 1 solution to the Kaup--Broer systems plotted at times $t=0$, $t=25$ and $t = 50$.}
\label{f14}
\end{figure}

\begin{figure}[h]
\center
\includegraphics[width=0.7 \textwidth]{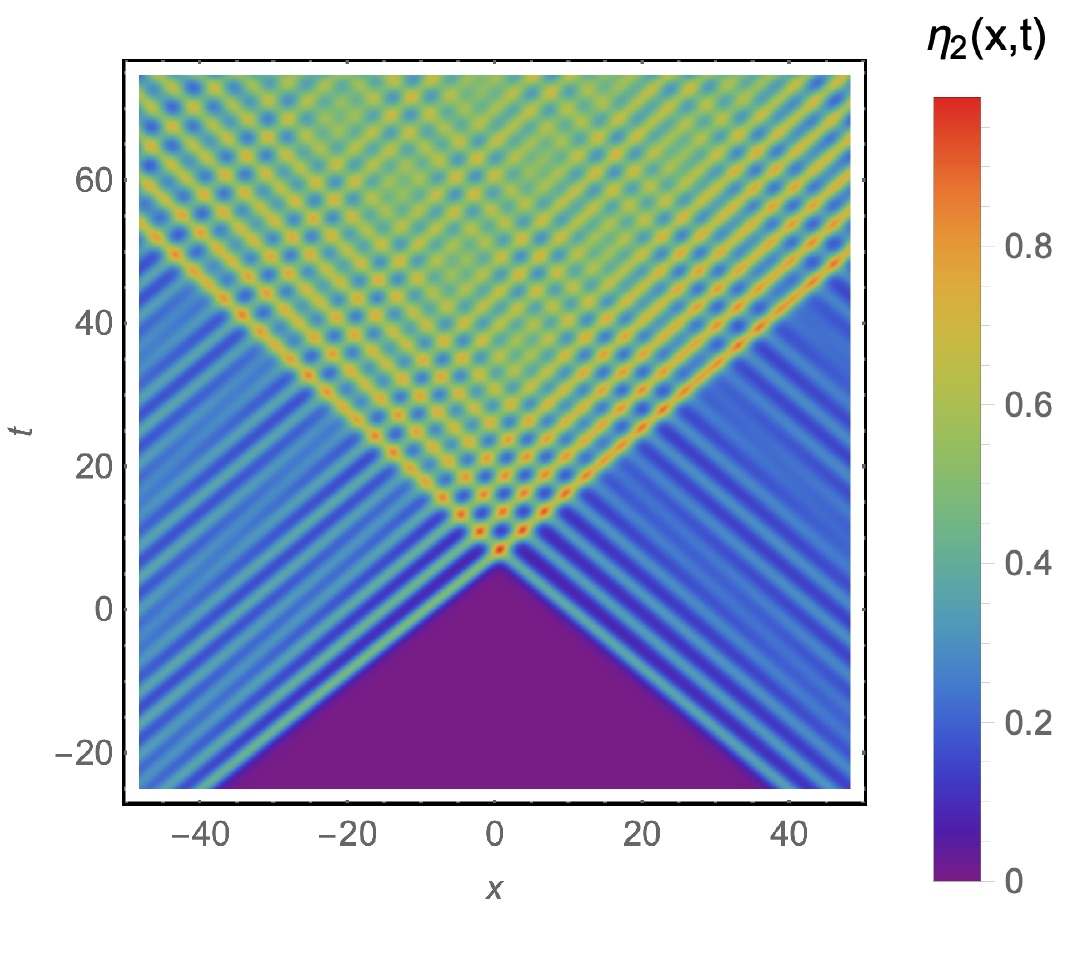}
\caption{A spacetime plot for the field $\eta_2$ for the case 2 solution to the Kaup--Broer system.}
\label{f21}
\end{figure}

\begin{figure}[h]
\center
\includegraphics[width=0.7 \textwidth]{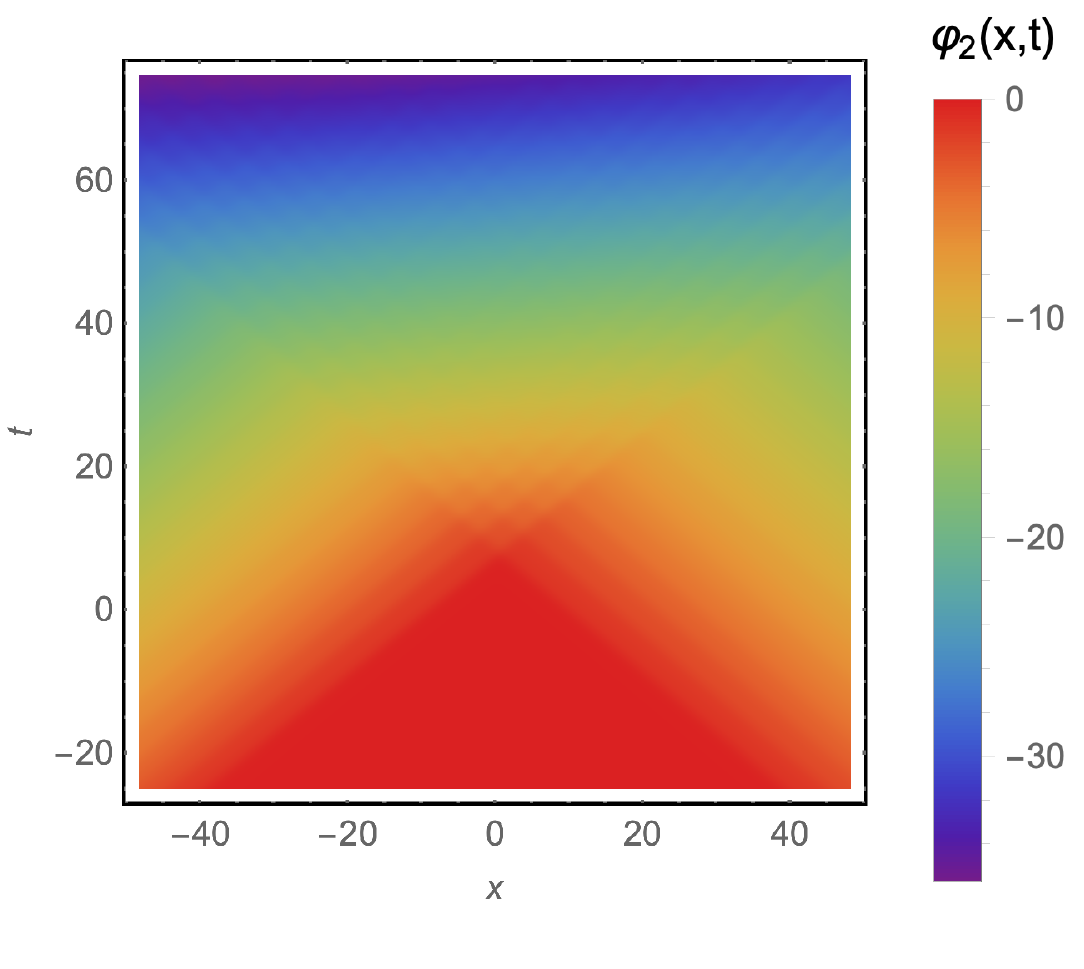}
\caption{A spacetime plot for the field $\varphi_2$ for the case 2 solution to the Kaup--Broer system.}
\label{f22}
\end{figure}

\begin{figure}[h]
\center
\includegraphics[width=0.65\textwidth]{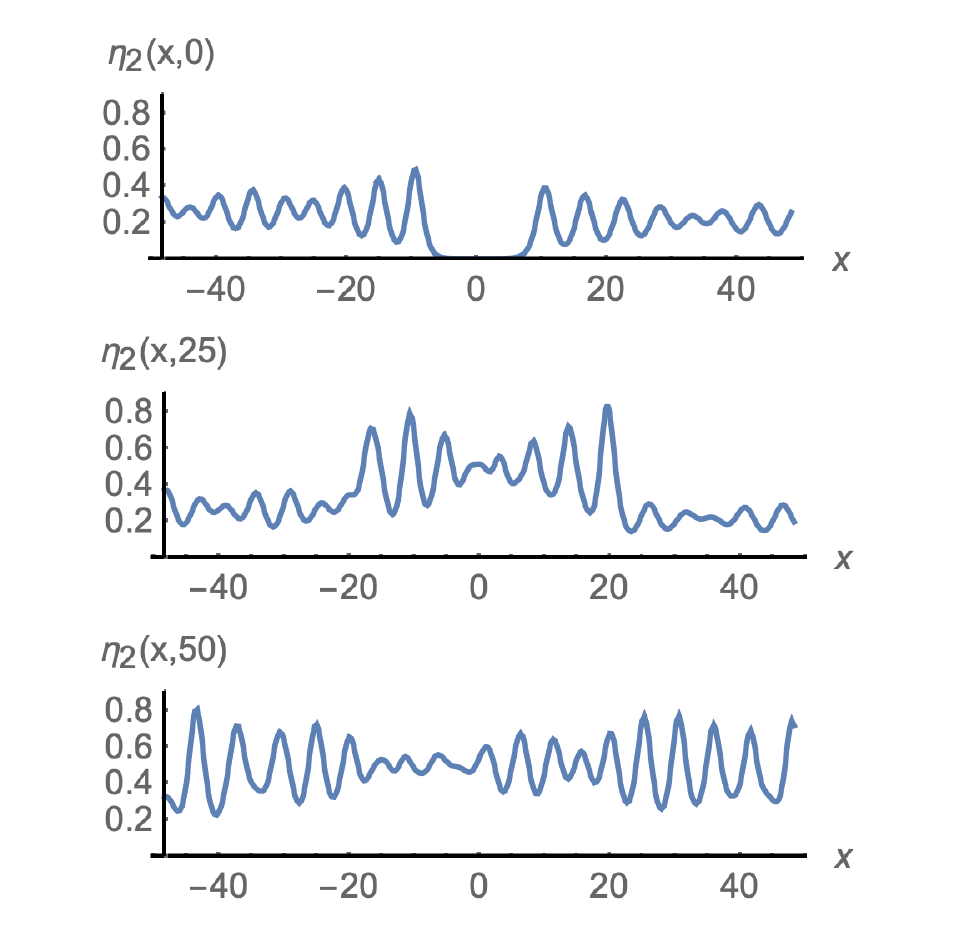}
\caption{Spatial plots for the field $\eta_2$ for the case 2 solution to the Kaup--Broer systems plotted at times $t=0$, $t=25$ and $t = 50$.}
\label{f23}
\end{figure}

\begin{figure}[h]
\center
\includegraphics[width=0.65 \textwidth]{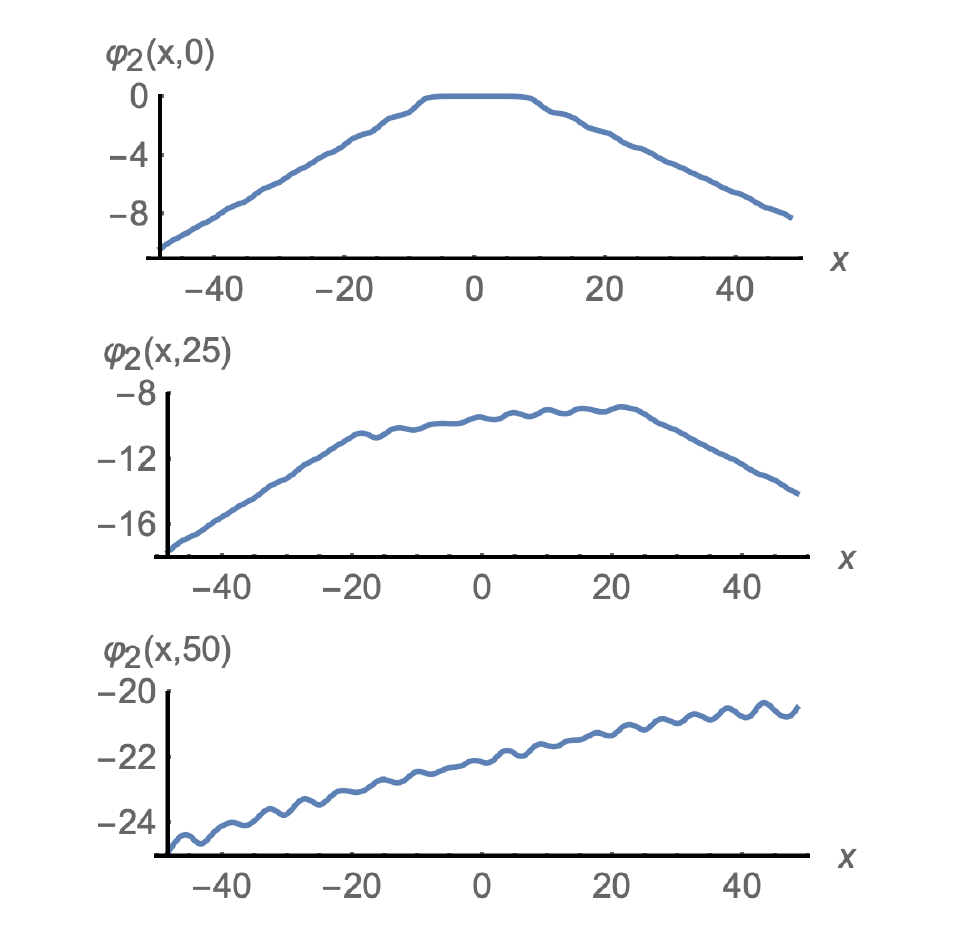}
\caption{Spatial plots for the field $\varphi_2$ for the case 2 solution to the Kaup--Broer systems plotted at times $t=0$, $t=25$ and $t = 50$.}
\label{f24}
\end{figure}

\begin{figure}[h]
\center
\includegraphics[width= 0.7\textwidth]{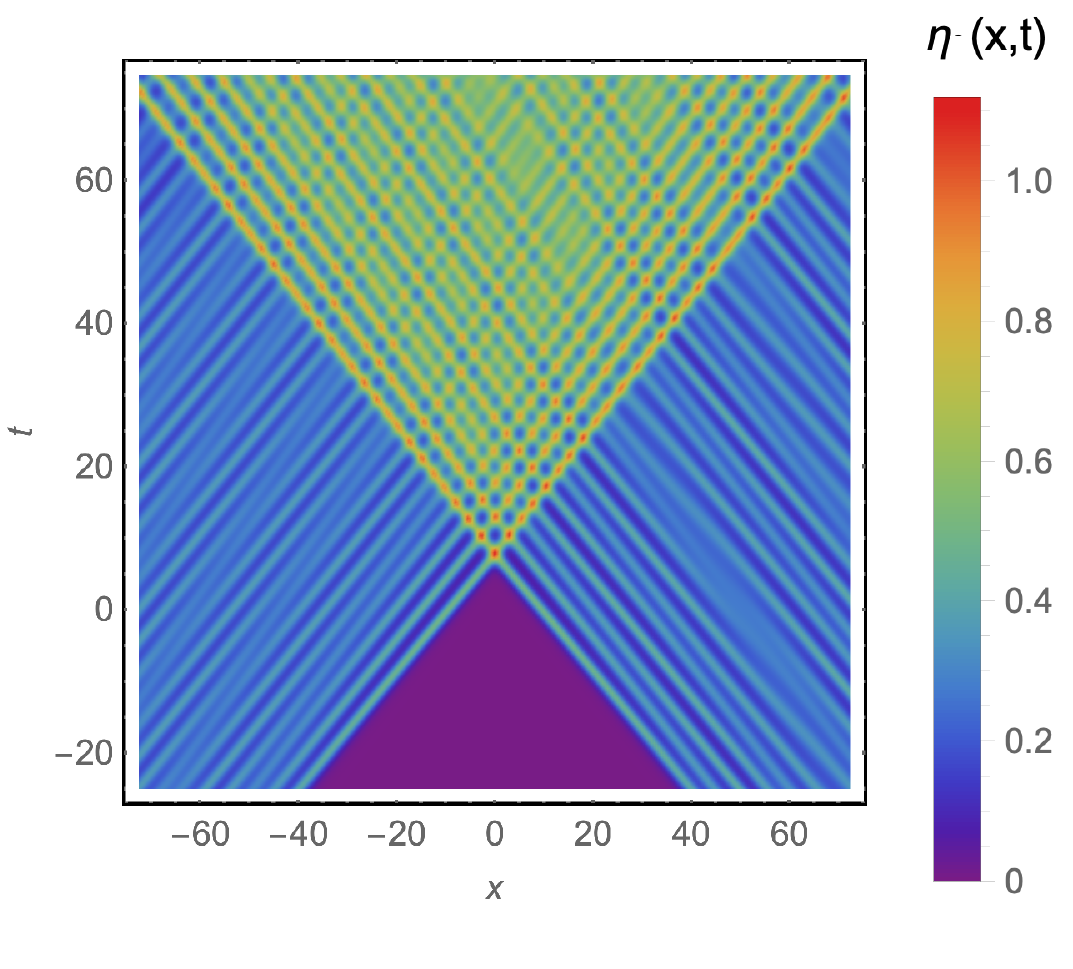}
\caption{A spacetime plot for the field $\eta_3$ for the case 3 solution to the Kaup--Broer system.}
\label{f31}
\end{figure}

\begin{figure}[h]
\center
\includegraphics[width= 0.7\textwidth]{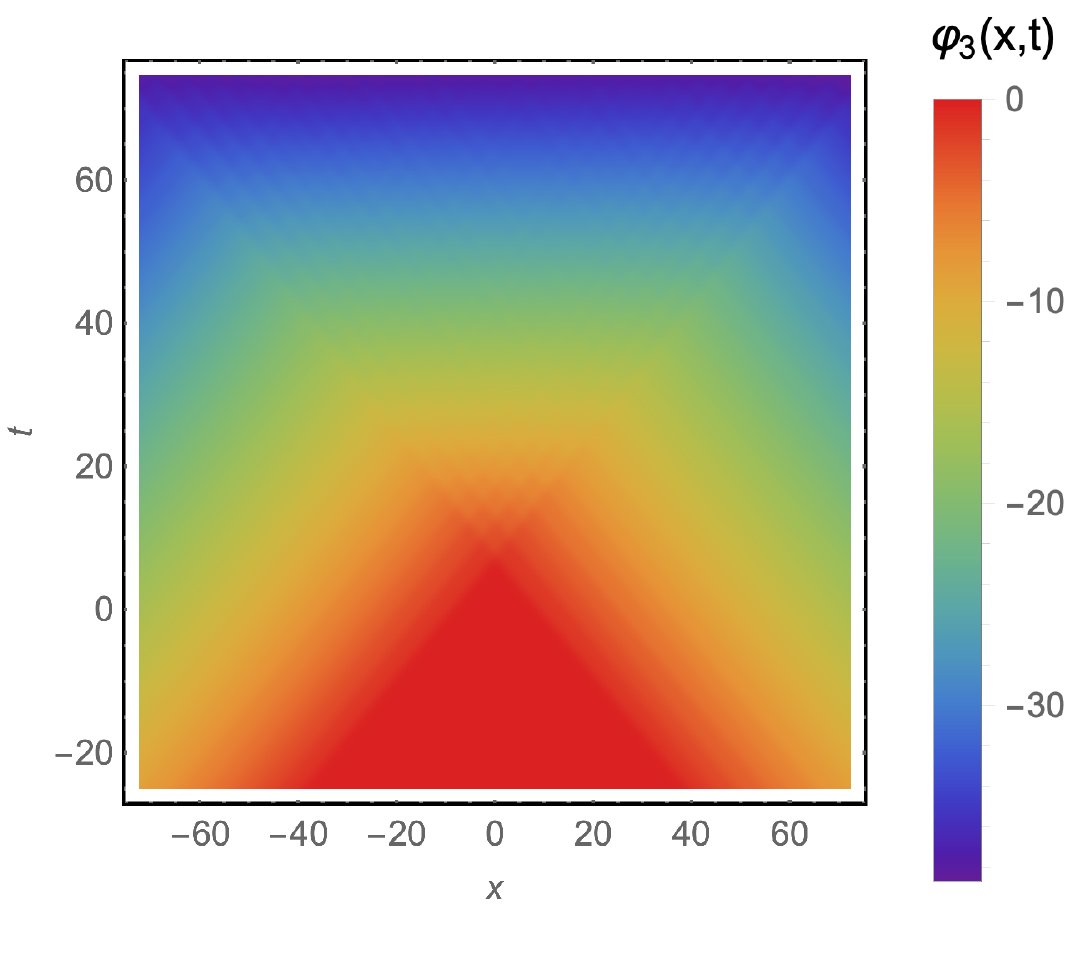}
\caption{A spacetime plot for the field $\varphi_3$ for the case 3 solution to the Kaup--Broer system.}
\label{f32}  
\end{figure}

\begin{figure}[h]
\center
\includegraphics[width= 0.65\textwidth]{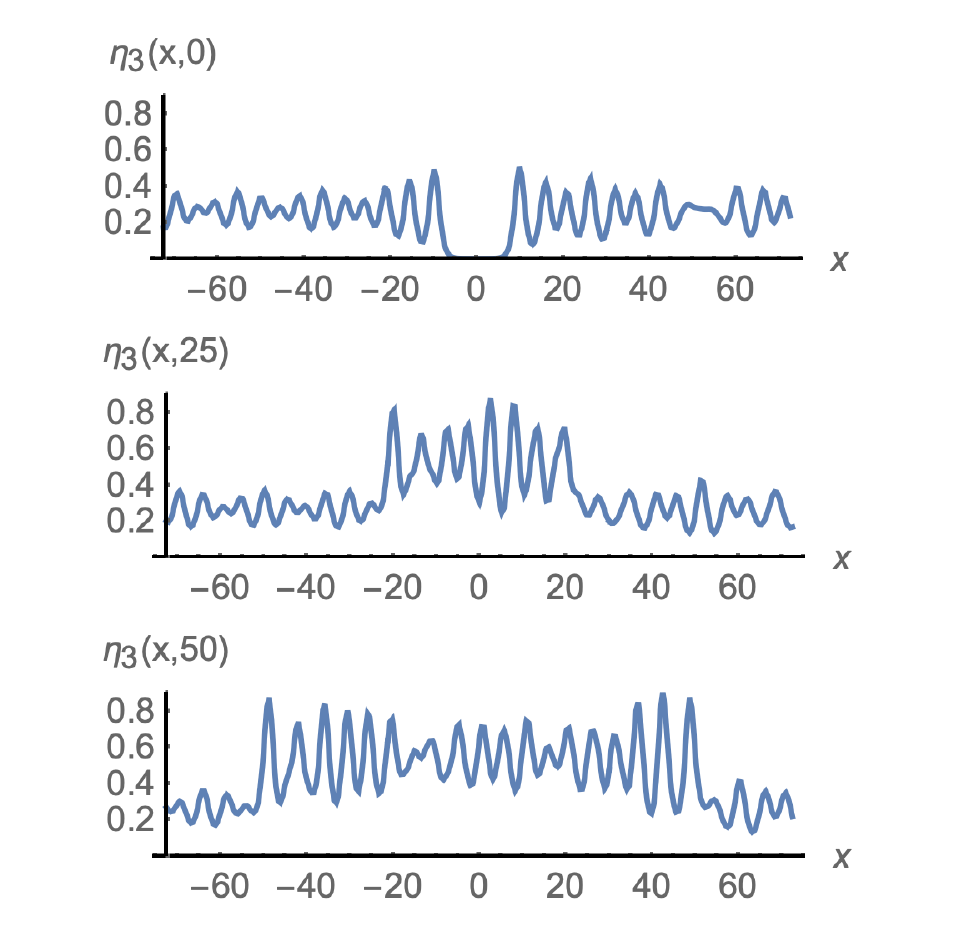}
\caption{Spatial plots for the field $\eta_3$ for the case 3 solution to the Kaup--Broer systems plotted at times $t=0$, $t=25$ and $t = 50$.}
\label{f33} 
\end{figure}

\begin{figure}[h]
\center
\includegraphics[width= 0.65\textwidth]{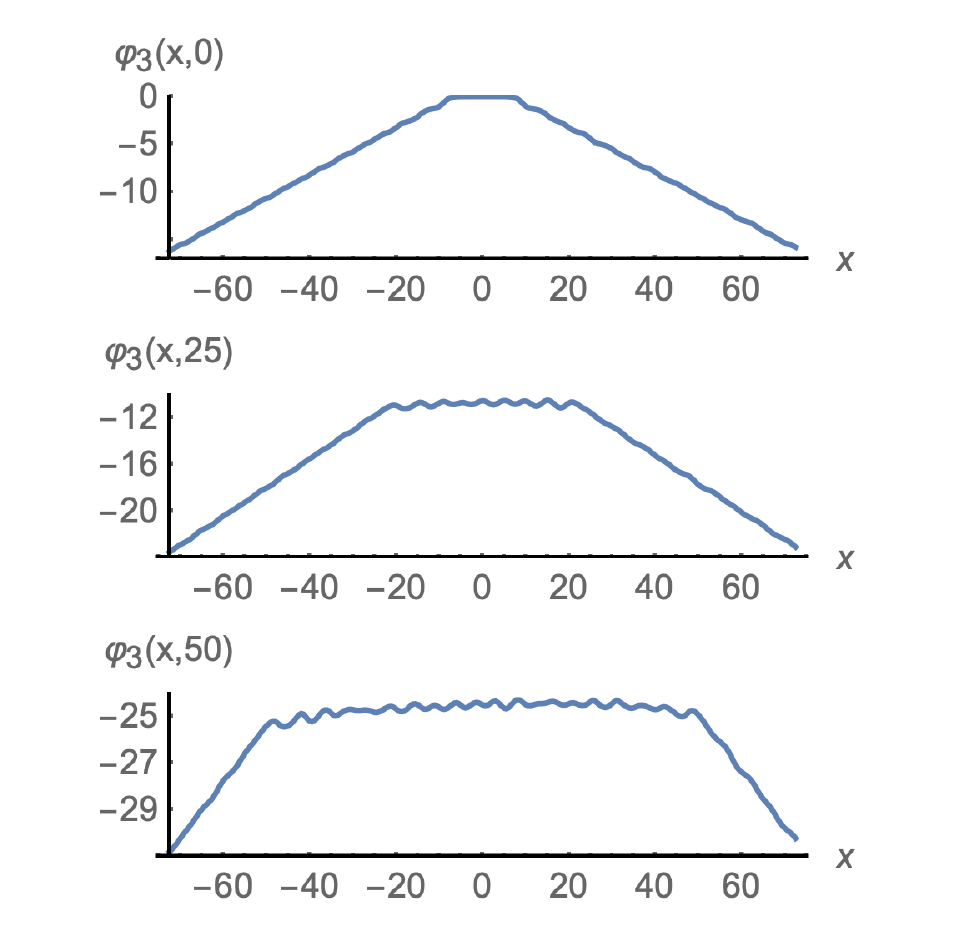}
\caption{Spatial plots for the field $\varphi_3$ for the case 3 solution to the Kaup--Broer systems plotted at times $t=0$, $t=25$ and $t = 50$.}
\label{f34}   
\end{figure}

\clearpage

\section{Conclusion}

In this paper we derived the $\bar \partial$ dressing method for the 4 cases of the Kaup--Broer system as dimension reductions of the complete complexification of a completely integrable 2+1 dimension generalization of the Kaup--Broer system.
We then applied the dressing method to compute examples of interesting solutions to the most common scaling of the Kaup--Broer system describing non-capillary waves with normal gravitation forcing.
However, this method also provides a starting place for the computation of many other solutions. 
A logical next step is the computations of exact real solutions to the other canonical scalings of the Kaup--Broer system.
We have also barely touched on exact solutions to the 2+1 dimensional generalization of the Kaup--Broer system, this gives much room for further research.
In particular, the solutions briefly discussed in subsection \ref{ssec:Nsolgen} should be analyzed with more detail.
In section 6 we made some conjectures on the structure of some numerical solutions related to the finite gap $g$-period solutions of Matveev and Yavor \cite{MY79}.
We believe that an approach similar to that Girotti, Grava and McLaughlin \cite{GGM18} can be applied to the Riemann--Hilbert problem discussed in section 4 of this paper.
The isomorphism of the Kaup--Broer system with a complexified NLS equation \cite{MY79} means that the construction of NLS primitive solutions should be similar.
Finally, as mentioned in the introduction, this method should be well suited to analysis of soliton gasses consisting of many counter propagating solitons.
These topics will be discussed in future papers.

\section{Acknowledgments}

The authors gratefully acknowledge the support of NSF grant DMS-1715323.


\begin{thebibliography}{00}



\bibitem{ABF83} M. J. Ablowitz, D. Bar-Yaacov, A. S. Fokas, {\it On the inverse scattering transform for the Kadomtsev--Petviashvili equation,} 	Stud. Appl. Math., 1983, Volume 69, 135-143.

\bibitem{ABM17} D.M. Ambrose, J.L. Bona, and T. Milgrom, {\it Global solutions and ill-posedness for the Kaup system and related Boussinesq systems}, Indiana U. Math. J., 2017.

\bibitem{B87} L. V. Bogdanov, {\it Veselov--Novikov equation as a natural two-dimensional generalization of the Korteweg--de Vries equation,} Theoret. and Math. Phys., 1987 , Volume 70, 219-223.

\bibitem{BM88} L. V. Bogdanov, S. V. Manakov, {\it The nonlocal $\bar \partial$-problem and (2+1) dimensional soliton equations,} Phys. A. and Math. Gen., 1988, Volume 21, L537-L544.

\bibitem{BZ02} L.V. Bogdanov, V.E. Zakharov, {\it The Boussinesq equation revisited,} Phys. D, 2002, Volume 165, No. 3-4, 137-162.

\bibitem{BBEIM94} E. D. Belokolos, A. I. Bobenko, V. Z. Enol?skii, A. R. Its, and V. B. Matveev, {\it Algebro-Geometric Approach to Nonlinear Integrable Equations}, Springer, 1994.

\bibitem{B75} L. P. F. Broer, {\it Approximate equations for long water waves,} Appl. Sci. Res., 1975, Volume 31, 377-295.

\bibitem{CY14} Y. Chen and H. Yeh, {\it Laboratory experiments on counter-propagating collisions of solitary waves. Part 1. Wave interactions}, J. Fluid Mech. 749, 577, 2014.

\bibitem{D81}{Dubrovin, B. A., 1981, {\it Theta functions and non-linear equations}, Russ. Math. Surv., 1981 Volume 36, 11-92.}

\bibitem{D09} B. A. Dubrovin. {\it Integrable Systems and Riemann Surfaces Lecture Notes}. http://people.sissa.it/{\textasciitilde}dubrovin/rsnleq\_web.pdf, 2009.

\bibitem{DZZ16} S. Dyachenko, D. Zakharov, V. Zakharov, {\it Primitive potentials and bounded solutions of the KdV equation,} Phys. D, 333, 148-156, 2016.

\bibitem{GGM18} M. Girotti, T. Grava, K. McLaughlin, {\it Rigorous asymptotics of a KdV soliton gas}, arXiv:1807.00608.

\bibitem{K75} D. J. Kaup, {\it A higher-order water-wave equation and a method for solving it,} 1975, Prog. of Theoret. Phys., Volume 54, No. 2, 396-408.

\bibitem{K86} B. A. Kupershmidt, {\it Mathematics of dispersive waves,} 1986, Commun. Math. Phys., Volume 99, 51-73.

\bibitem{M76} S. V. Manakov, {\it The inverse scattering method and two-dimensional evolution
  equations,} Uspekhi Mat. Nauk, 1976, Volume 31, 245-246.
  
\bibitem{MY79} V. B. Matveev, M. I. Yavor, {\it Solutions presque periodiques et a n-solitons de pequation hydrodynamique non lineaire de Kaup,} 1979, Ann. Inst. H. Poincare, Sect. A XXXI, No. 1, 25-41.
 
\bibitem{N18} P. Nabelek, {\it Applications of Complex Variables to Spectral Theory and Completely Integrable Partial Differential Equations,} PhD diss., University of Arizona, 2018.

\bibitem{N19} P. Nabelek, {\it Algebro-Geometric Finite Gap Solutions to the Korteweg--de Vries Equation as Primitive Solutions}, submitted to Physica D. 2019 (arXiv:1907.09667)
 
\bibitem{NZZ19} P. Nabelek, D. Zakharov, V. Zakharov, {\it On symmetric primitive potentials}, J. Int. Sys. Volume 4, Issue 1, 2019, xyz006, https://doi.org/10.1093/integr/xyz006.
  
\bibitem{N83}  S. P. Novikov, {\it Two-dimensiobal Schr\"{o}dinger Operators in Periodic Fileds,} Itogi Nauki i Tekhniki, Ser. Sovrem. Prob. Mat., 1983, Volume 23, 3-33.  

\bibitem{NV84}  S. P. Novikov and A. P. Veselov, {\it Finite-zone, two-dimensional, potential Schrodinger operators, explicit formula and evolution equations,} Sov. Math. Dokl., 1984, Volume 30, 288-291.

\bibitem{RNMOM19} I. Redor, E. Barth\'elemy, H. Michallet, M. Onorato, and N. Mordant, {\it Experimental Evidence of a Hydrodynamic Soliton Gas}, Phys. Rev. Lett., Volume 122, 214502, 2019.

\bibitem{RP11} C. Rogers, O. Pashaev, {\it On a 2+1-Dimensional Whitham--Broer--Kaup System: A Resonant NLS Connection,} Stud. Appl. Math., August 2011, Volume 127 Issue 2, 141-152.

\bibitem{TD13} Trogdon, T. and B. Deconinck, ?A Riemann-Hilbert problem for the finite-genus solutions of the KdV equation and its numerical solution,? {\it Phys. D}, 2013, Volume 251, 1-18.

\bibitem{S86} Smirnov, A.O., {\it Real finite-gap regular solutions of the Kaup-Boussinesq equation}, Theor. Math. Phys. Volume 66, 19, 1986, https://doi.org/10.1007/BF01028935

\bibitem{W13} A. Wazwaz, {\it Multiple soliton solutions for the Whitham--€"Broer--"Kaup model in the shallow water small-amplitude regime,} Phys. Scr., 2013, Volume 88, No. 3, 035007.

\bibitem{Z89} V. E. Zakharov, On the dressing method, in: P. C. Sabatier (Ed.), Inverse Methods in Action, Springer Verlag, 1989, 602-623.

\bibitem{ZM84} V. E. Zakharov, S. V. Manakov, {\it The many-dimensional integrable systems and their solutions,} Zap. Nauch. Sem. LOMI, 1984, Volume 133, 11-25.

\bibitem{ZM85} V. E. Zakharov, S. V. Manakov, Construction of higher-dimensional nonlinear integrable systems and of their solutions, Funct. Anal. Appl., 1985, Volume 19, 89-101.

\bibitem{ZDZ16} D. Zakharov, S. Dyachenko, V. Zakharov, {\it Bounded solutions of KdV and non-periodic one-gap potentials in quantum mechanics}, Lett. Math. Phys. 106 (2016) no. 6, 731-740

\bibitem{ZZD16} D. Zakharov, V. Zakharov, S. Dyachenko, {\it Non-periodic one-dimensional ideal conductors and integrable turbulence,} Phys. Lett. A 380, no. 46, 3881-3885, 2016


\end{thebibliography}
\end{document}